\documentclass[onecolumn, 12pt]{IEEETran}
\usepackage{graphicx}
\usepackage{cite}
\usepackage{amsmath, amsfonts, amssymb}
\usepackage{subfigure, amsthm}
\usepackage{comment}
\usepackage{array}

\usepackage[dvipsnames]{xcolor}

\newcommand{\PR}[1]{\mathrm{Pr}[#1]}
\newcommand{\NN}{\nonumber}
\newcommand{\NNL}{\nonumber\\}
\newcommand{\RBob}{\mathcal{R}_{\textrm{Bob}}}
\newcommand{\REve}{\mathcal{C}_{\textrm{Eve}}}
\newcommand{\DoF}[1]{\lim_{P\to\infty}\frac{#1}{\log_2 P}}
\newcommand{\tDoF}[1]{\underset{P\to\infty}{\lim}\left(#1 / \log_2 P\right)}
\newcommand{\Span}[1]{\mathrm{span}(#1)}
\newcommand{\Chord}[1]{d_c(#1)}
\newcommand{\Chords}[1]{d_c^2(#1)}

\newcommand{\Grassmann}[1]{\mathcal{G}_{#1}(\mathbb{C})}
\newcommand{\Avr}[1]{\mathbb{E}\left[#1\right]}

\newcommand{\JHcomment}[1]{}%Bittersweet

\newtheorem{theorem}{Theorem}
\newtheorem{lemma}{Lemma}
\newtheorem{remark}{Remark}

\theoremstyle{definition}

\title{Multiuser Diversity for Secrecy Communications Using Opportunistic Jammer Selection  --
     Secure DoF and Jammer Scaling Law}
\author{\IEEEauthorblockN{Jung~Hoon~Lee},~\IEEEmembership{Student Member,~IEEE},
     and \IEEEauthorblockN{Wan~Choi},~\IEEEmembership{Senior Member,~IEEE}
\thanks{A part of this paper has been presented in \textit{ IEEE Global Commun. Conf. (Globecom)}, Anaheim, CA,
Dec. 2012.}
\thanks{J.~H.~Lee and W.~Choi are with Department of Electrical
    Engineering, Korea Advanced Institute of Science and Technology
    (KAIST), Daejeon 305-701, Korea (e-mail: tantheta@kaist.ac.kr
    wchoi@ee.kaist.ac.kr). }% <-this % stops a space
}

\begin{document}
\maketitle

\begin{abstract}
In this paper, we propose opportunistic jammer selection in a
wireless security system for increasing the secure degrees of
freedom (DoF) between a transmitter and a legitimate receiver (say,
Alice and Bob). There is a jammer group consisting of $S$ jammers
among which Bob selects $K$ jammers. The selected jammers transmit
independent and identically distributed Gaussian signals to hinder
the eavesdropper (Eve).
Since the channels of Bob and Eve are independent, we can select the
jammers whose jamming channels are aligned at Bob, but not at Eve.
As a result, Eve cannot obtain any DoF unless it has more than
$KN_j$ receive antennas, where $N_j$ is the number of jammer's
transmit antenna each, and hence $KN_j$ can be regarded as
defensible dimensions against Eve. For the jamming signal alignment
at Bob, we propose two opportunistic jammer selection schemes and
find the scaling law of the required number of jammers for target
secure DoF by a geometrical interpretation of the received signals.

\end{abstract}

\begin{IEEEkeywords}
Physical layer security, secure DoF, multiuser diversity,
opportunistic jammer selection
\end{IEEEkeywords}

\clearpage

%%%%%%%%%%%%%%%%%%%%%%%%%%%%%%%%%%%%%%%%%%%%%%%%%%%%%%%%%%%%%%%%%%%
% % % % % % % % % % % % % % % % % % % % % % % % % % % % % % % % % %
%%%%%%%%%%%%%%%%%%%%%%%%%%%%%%%%%%%%%%%%%%%%%%%%%%%%%%%%%%%%%%%%%%%
\section{Introduction}

Wireless communication systems are vulnerable to eavesdropping
because transmitted signals are broadcasted over the air. Privacy
and security start being treated as important issues.
%
%in wireless communication systems.
%
Traditionally, security in wireless communications has mainly been
addressed in upper layers and focused on computational security such
as cryptography.
Recently, information-theoretic security in physical layer has
received great attentions because it enables perfect secrecy without
the aid of encryption/decryption keys.

%%%%%%%%%%%%%%%%%%%%%%%%%%%%%%%%%%%%%%%%%%%%%%%%%%%%%%%%%%%%%%%%%%%

A typical model of a security communication system consists of three
nodes -- a transmitter (Alice), a legitimate receiver (Bob), and a
passive eavesdropper (Eve). Information theoretic security study was
opened by Shannon \cite{S1949} with the notion of perfect secrecy.
The wiretap channel model was first introduced by Wyner \cite{W1975}
and \emph{secrecy capacity} is defined in \cite{CK1978} as the
maximum achievable rate of Bob preventing Eve from obtaining any
information.
Positive secrecy capacity was shown to be achieved when Bob's
channel is less noisy than Eve's.
For guaranteeing positive secrecy rate, artificial noise was
additionally transmitted from the transmitter \cite{GN2008}.
%
%When Alice has more antennas than Eve, the secrecy capacity can be
%achieved by sending artificial noise over the additional antenna
%dimensions. If Alice has a single antenna, multiple jamming relays
%can be exploited for the same purpose.
%
%\textbf{In \cite{ZMT2010}, artificial noise is transmitted in a
%symmetric interference channel where two transmitter-receiver pairs
%do not trust each other. It was shown that artificial noise can
%enlarge the secrecy rate region.}
%
The authors in \cite{ZMT2010} showed that the artificial noise can
enlarge the secrecy rate region in a symmetric interference channel.
The secrecy issues were also studied in cooperative relay systems
\cite{VAG2009, CZS2011, DLG2011}.
%
%The authors in \cite{VAG2009} investigated multi-user diversity in
%opportunistic relaying for improving the secrecy outage performance.
%
In \cite{CZS2011, DLG2011}, joint selection of relays and jammers
was studied and opportunistic jamming and relay chatting was
proposed.

\begin{comment}
%
Multi-user diversity obtained from opportunistic relaying was
exploited to improve secrecy outage in \cite{VAG2009}.
%
Joint relay and jammer selection in relay networks was studied in
\cite{KTM2009, CZS2011}, and  opportunistic cooperative jamming and
relay chatting was proposed in \cite{DLG2011}.
\end{comment}

%%%%%%%%%%%%%%%%%%%%%%%%%%%%%%%%%%%%%%%%%%%%%%%%%%%%%%%%%%%%%%%%%%%

As an alternative measure to secrecy capacity, secure degrees of
freedom (DoF) have popularly been investigated.
%
%in many literatures.
%
In \cite{GJ2008}, the secure DoF was found in an X network. For
$N$-user Gaussian interference channels, an interference channel
with confidential message and an interference channel with an
external eavesdropper were studied in \cite{KGL2008, KGL2011},
respectively. The secure DoFs for those cases were shown to be
$\frac{N-2}{2N-2}$ and $\frac{N-2}{2N}$, respectively, while a half
DoF per orthogonal dimension is to be achievable via the
\emph{interference alignment} (IA) scheme \cite{JS2008, CJ2008} in
the absence of the secrecy constraint.

The multiuser diversity, exploiting multiuser dimensions by serving
the selected users with good channel conditions, can enhance the
performance of wireless communication systems \cite{VTL2002, ZD2007,
SH2005, HR2013, HC2011, CA2008}. In many cases, the multiuser
diversity asymptotically achieves the optimal performance with
considerably reduced channel feedback overhead.
Recently, the opportunistic interference alignment (OIA) has been
proposed by the authors in order to resolve the difficulties of IA
implementation \cite{LC2011_3, LC2013, LC2013_2}. In the OIA scheme,
user dimensions are exploited to align interfering channels; each
transmitter has multiple users and selects a single user having the
most aligned interfering channels. The OIA scheme does not require
the global channel state information (CSI) and large computational
complexity for precoding/postprocessing design. A bit surprisingly,
the achieved DoF by OIA was shown to be higher than that of
conventional IA thanks to multiuser dimensions. This result
motivates new applications using the concept of OIA for security
communication systems.

%%%%%%%%%%%%%%%%%%%%%%%%%%%%%%%%%%%%%%%%%%%%%%%%%%%%%%%%%%%%%%%%%%%

In this paper, we propose opportunistic jammer selection (OJS)
schemes for secure DoF. The basic idea of OJS is to align jamming
signals at Bob's receiver via jammer selection, while the jamming
signals are not aligned at Eve. There is a jammer group consisting
of $S$ jammers, and Bob selects $K(\ge2)$ jammers among them whose
jamming signals are most aligned.
Since we use jamming subspaces formed by jammers, the selected
jammers simply use independent and identically distributed (i.i.d)
Gaussian signals independent of Eve's CSI as well as Alice's.
Denoting the numbers of antennas at Alice, Bob, and Eve by $N_t$,
$N_r$, and $N_e$, respectively, we focus on the cases that
$N_t+N_j\le N_r < N_t + KN_j$ because these cases require the
jamming signal alignment at Bob's receiver.

Although the jamming signals from the selected jammers can be
aligned at Bob's receiver, they are randomly given to Eve and hence
span $KN_j$-dimensional subspace at Eve's receiver. Because Eve
requires more than $KN_j$ receive antennas to obtain non-zero
DoF\footnote{Eve's DoF is defined by $\tDoF{\REve}$ where $\REve$ is
Eve's channel capacity.}, we can regard $KN_j$ as \emph{defensible
dimensions} of the system.
For any number of Eve's receive antennas, we can make Eve's DoF zero
by increasing either the number of selected jammers $K$ or the
number of each jammer's antennas $N_j$.
On the contrary, the jamming signals from the selected jammers can
be aligned in $N_j$-dimensional subspace at Bob's receiver if the
number of jammers goes to infinity. In this case, Bob needs only
$N_t+N_j$ receive antennas to obtain DoF of $N_t$ while Eve needs
$N_t + KN_j$ receive antennas.

Although both OIA and OJS exploit multiuser (or multiple jammer)
dimensions to align interfering (or jamming) signals, the problems
are quite different. In OIA, each user in $N$-transmitter IC has
$N-1$ interfering channels jointly determined for each channel
realization. Thus, the OIA problem is to measure how much the $N-1$
interfering channels are aligned at the selected user when each user
has $N-1$ random interfering channels.
However, because Bob chooses $K$ jammers among total $S$ jammers in
OJS, the problem is changed to how much we can align the $K$ jamming
signals by directly picking each jamming signal among $S$ jamming
signals.

%%%%%%%%%%%%%%%%%%%%%%%%%%%%%%%%%%%%%%%%%%%%%%%%%%%%%%%%%%%%%%%%%%%
%\newpage
%%%%%%%%%%%%%%%%%%%%%%%%%%%%%%%%%%%%%%%%%%%%%%%%%%%%%%%%%%%%%%%%%%%

There have been many studies on utilizing jammers for increasing
secure DoF \cite{HY2008, BU2012, XU2012}.
In \cite{HY2008}, cooperative jamming with structured jamming
signals based on lattice coding was proposed for a Gaussian wiretap
channel. In multiple access fading channels, \cite{BU2012} proposed
two ergodic alignment schemes for secrecy communications -- scaling
based alignment (SBA) and ergodic secret alignment (ESA). In SBA,
users repeatedly transmit their symbols with proper scaling over
several consecutive time slots. Then, the signals are aligned at
eavesdropper's receiver and hence Eve's DoF becomes zero. In ESA,
the concept of ergodic IA \cite{NGJV2009} was extended to secrecy
communications. In \cite{XU2012}, cooperative jammers were exploited
to increase secure DoF in a Gaussian wiretap channel, where
transmitter and cooperative jammers send jointly designed signals
according to channel conditions.
Contrary to the previous works, in our proposed OJS schemes, the
selected jammers adopt i.i.d Gaussian signals oblivious to CSI.
Furthermore, any precoding design at Alice is not required for the
proposed scheme. Because the role of CSI is only for constructing
the wiretap code, the proposed OJS scheme does not suffer from
secure DoF degradation even with only Bob's CSI, which is a
practical advantage of the proposed OJS scheme.

%%%%%%%%%%%%%%%%%%%%%%%%%%%%%%%%%%%%%%%%%%%%%%%%%%%%%%%%%%%%%%%%%%%
% \newpage
%%%%%%%%%%%%%%%%%%%%%%%%%%%%%%%%%%%%%%%%%%%%%%%%%%%%%%%%%%%%%%%%%%%

Our contributions are summarized below.
\begin{itemize}
\item  Using a geometrical interpretation of jamming signals,
we define an \emph{alignment measure} representing how well the
selected jamming signals are aligned at Bob's receiver and quantify
the achievable alignment measure via jammer selection for the total
number of jammers.

\item To obtain the secure DoF, we propose two jammer selection
schemes -- the minimum DoF loss jammer selection and the
subspace-based jammer selection. Using the proposed jammer selection
schemes, we show that Bob can achieve the secure DoF of $d \in
(0,N_t]$ when the number of jammers is scaled by $S \propto
P^{dN_j/2}$ where $P$ is transmit power.

%\item \textbf{We show that the
%secure DoF of $d\in [0,N_t]$ can be achieved by via OJS with jammer
%scaling law $S=\BigO{P^{d/2}}$ invariant to the number of Eve's
%antennas. In this case, each jammer should have a single antenna and
%$(N_t+1)$ antennas are enough for Bob's receiver. Also, at least
%$N_e$ jammers should be selected to make Eve's achievable DoF zero.}

\item We show that the same secure DoF is achievable with secrecy
outage probability $\epsilon$ by the proposed OJS with only Bob's
CSI.

\end{itemize}

%%%%%%%%%%%%%%%%%%%%%%%%%%%%%%%%%%%%%%%%%%%%%%%%%%%%%%%%%%%%%%%%%%%
%\JHbox{Organization}

The rest of the paper is organized as follows.
Section II presents the system model. The achievable secure DoF via
opportunistic jammer selection is analyzed in Section III. We
geometrically interpret the jamming channels in Section IV, and
propose opportunistic jammer selection schemes in Section V. In
Section VI, we derive the scaling laws of the required number of
jammers for target secure DoF. Numerical results are presented in
Section VII, followed by concluding remarks in Section VIII.

%%%%%%%%%%%%%%%%%%%%%%%%%%%%%%%%%%%%%%%%%%%%%%%%%%%%%%%%%%%%%%%%%%%
% % % % % % % % % % % % % % % % % % % % % % % % % % % % % % % % % %
%%%%%%%%%%%%%%%%%%%%%%%%%%%%%%%%%%%%%%%%%%%%%%%%%%%%%%%%%%%%%%%%%%%
\section{System Model}

Our system model is depicted in Fig. \ref{fig:system_model}. Alice
wants to send secret messages to Bob, and a passive eavesdropper,
Eve, overhears the secret messages.
The numbers of antennas at Alice, Bob, and Eve are denoted by $N_t$,
$N_r$, and $N_e$, respectively.
To prevent Eve's eavesdropping, there is a jammer group consisting
of $S$ jammers with $N_j$ antennas each, and Bob selects $K~(K\ge2)$
jammers among them.
The selected jammers simultaneously transmit i.i.d. Gaussian signals
independent of Alice's message so that the jamming signals interfere
with Bob as well as Eve.

We assume that Alice and the selected jammers fully utilize their
antenna dimensions, i.e., Alice and each selected jammer transmit
$N_t$ and $N_j$ streams, respectively.
We also assume that Eve has more antennas than each jammer, i.e.,
$N_e>N_j$, and Bob selects $K$ jammers such that $K N_j \ge N_e$. We
consider the cases that Bob has $N_r$ receive antennas such that
$N_t+N_j \le N_r < N_t + KN_j$.
If Bob has a less number of receive antennas than the total number
of antennas at Alice and a jammer, i.e., $N_r<N_t + N_j$, Bob cannot
obtain DoF of $N_t$; otherwise, if the number of Bob's antennas is
larger than or equal to the total number of Alice and all selected
jammers' antennas, i.e., $N_r\ge N_t+KN_j$, Bob can easily achieve
DoF of $N_t$ with zero-forcing like schemes.

Since Bob knows the jamming channel from each jammer, we assume that
Bob selects $K$ jammers in the jammer group with only its own CSI.
It is also assume that Eve has its own CSI from the selected
jammers, which is independent of Bob's CSI.
For wiretap code construction at Alice, we firstly assume that Alice
knows Bob's achievable rate and Eve's channel capacity after Bob's
jammer selection.
Then, in Section \ref{section:partial_CSI} we show that the required
jammer scaling for the target secure DoF is the same even in
practical scenarios that Alice has no information about Eve.

In this paper, we adopt the quasi-static fading channel model
\cite{BB2011}, where the coherent interval is longer than the jammer
selection procedure and the length of a codeword. That is, the
channel coefficients remain constant over the transmission of an
entire codeword but each codeword suffers from an independent
channel.

Let $s_1, \ldots, s_K$ be the indices of the $K$ selected jammers in
the jammer group. Then, the received signal of Bob, $\mathbf{y} \in
\mathbb{C} ^{N_r\times 1}$, is given by
\begin{align}
 \mathbf{y} = \mathbf{H}_0 \mathbf{x}_0 +
    \sum_{k=1}^K \mathbf{H}_{s_k} \mathbf{x}_{s_k}
    + \mathbf{n}, \label{eqn:y_Bob}
\end{align}
where $\mathbf{H}_0 \in \mathbb{C}^{N_r\times N_t}$ is the channel
matrix from Alice to Bob, and $\mathbf{H}_{s_k} \in \mathbb{C}^{N_r
\times N_j}$ is the channel matrix from the $s_k$th jammer to Bob.
We assume that all elements of the channel matrices are i.i.d.
complex Gaussian random variables with zero mean and unit variance.
The vectors $\mathbf{x}_0 \in \mathbb{C}^{N_t\times 1}$ and
$\mathbf{x}_{s_k} \in \mathbb{C}^{N_j\times 1}$ are the transmit
signal from Alice and the $s_k$th jammer, respectively, satisfying
$\mathbb{E} [\mathbf{x}_0 \mathbf{x}_0^ \dagger] = (P/N_t)
\mathbf{I}_{N_t}$ and $\mathbb{E} [\mathbf{x}_{s_k}
\mathbf{x}_{s_k}^\dagger] = (P/N_j) \mathbf{I}_{N_j}$, where
$(\cdot)^\dagger$ denotes conjugate transpose, $P$ is total transmit
power budget at each node, and $\mathbf{I}_{N_t}$ is an $N_t\times
N_t$ identity matrix.
$\mathbf{n}\in \mathbb{C}^{N_r\times 1}$ is a circularly symmetric
complex Gaussian noise vector such that $\mathbf{n} \sim
\mathcal{CN}(0, \mathbf{I}_{N_r})$.

From \eqref{eqn:y_Bob}, the capacity of Bob becomes
\begin{align}
 \mathcal{C}_{\textrm{Bob}} \triangleq \log_2 \left\vert
 \mathbf{I}_{N_r} +
     \frac{P}{N_t}\mathbf{H}_0 \mathbf{H}_0 ^\dagger
    \left(\mathbf{I}_{N_r} + \frac{P}{N_j}
        \sum_{k=1}^K \mathbf{H}_{s_k}\mathbf{H}_{s_k} ^\dagger
    \right)^{-1} \right\vert. \NN
\end{align}
In this paper, we assume that Bob uses the postprocessing matrix
$\mathbf{V}^\dagger \in \mathbb{C} ^{N_t \times N_r}$ such that
$\mathbf{V} ^\dagger \mathbf{V} = \mathbf{I}_{N_t}$ to suppress the
jamming signals. After postprocessing, the received signal at Bob is
given by
\begin{align}
 \mathbf{V}^\dagger\mathbf{y} = \mathbf{V}^\dagger\mathbf{H}_0 \mathbf{x}_0 +
 \sum_{k=1}^K \mathbf{V}^\dagger\mathbf{H}_{s_k} \mathbf{x}_{s_k} +
 \mathbf{V}^\dagger\mathbf{n},\NN
\end{align}
and the achievable rate of Bob becomes
\begin{align}
 \RBob &\triangleq \log_2
    \frac{\left\vert\mathbf{I}_{N_t} + \mathbf{V}^\dagger
    \left( \frac{P}{N_t}\mathbf{H}_0 \mathbf{H}_0 ^\dagger
    + \sum_{k=1}^K \frac{P}{N_j}\mathbf{H}_{s_k} \mathbf{H}_{s_k} ^\dagger
    \right)\mathbf{V}\right\vert}
    {\left\vert\mathbf{I}_{N_t} + \frac{P}{N_j} \sum_{k=1}^K \mathbf{V}^\dagger \mathbf{H}_{s_k}
        \mathbf{H}_{s_k} ^\dagger \mathbf{V}\right\vert}.\NN
% \label{eqn:R_Bob}
\end{align}

On the other hand, the received signal at Eve denoted by
$\bar{\mathbf{y}}\in\mathbb{C}^{N_e\times 1}$ becomes
\begin{align}
 \bar{\mathbf{y}} = \mathbf{G}_0 \mathbf{x}_0 +
 \sum_{k=1}^K \mathbf{G}_{s_k} \mathbf{x}_{s_k} +
 \bar{\mathbf{n}},\NN
% \label{eqn:y_Eve}
\end{align}
where $\mathbf{G}_0 \in \mathbb{C}^{N_e\times N_t}$ and
$\mathbf{G}_{s_k} \in \mathbb{C}^{N_e\times N_j}$
are the channel matrices from Alice and from the selected jammer
$s_k$, respectively, whose elements are i.i.d. complex Gaussian
random variables with zero mean and unit variance.
Also, $\bar{\mathbf{n}}\in \mathbb{C}^{N_e\times 1}$ is a circularly
symmetric complex Gaussian noise vector such that $\bar{\mathbf{n}}
\sim \mathcal{CN}(0, \mathbf{I}_{N_e})$.
Thus, the channel capacity of Eve is given by
\begin{align}
 \REve &\triangleq \log_2 \left\vert
 \mathbf{I}_{N_e} +
     \frac{P}{N_t}\mathbf{G}_0 \mathbf{G}_0 ^\dagger
    \left(\mathbf{I}_{N_e} + \frac{P}{N_j}
        \sum_{k=1}^K \mathbf{G}_{s_k}\mathbf{G}_{s_k} ^\dagger
    \right)^{-1} \right\vert \label{eqn:R_Eve}.
\end{align}
Therefore, the secrecy rate $[\RBob-\REve]^+$ is achievable at Bob
for each channel realization through the Wyner's encoding scheme
\cite{LPS2009, W1975} with nested code structure, where $[\cdot]^+$
denotes $\max(\cdot, 0)$. In an average sense, we obtain the secrecy
rate and the secure DoF given, respectively, by
\begin{align}
 &\textrm{Secrecy rate} = \mathbb{E}\left\{[\RBob - \REve]^+\right\}, \NNL
 &\textrm{Secure DoF} = \mathbb{E}\left\{\DoF{[\RBob - \REve]^+}\right\}.
 \label{eqn:SDoF}
\end{align}

%%%%%%%%%%%%%%%%%%%%%%%%%%%%%%%%%%%%%%%%%%%%%%%%%%%%%%%%%%%%%%%%%%%
% % % % % % % % % % % % % % % % % % % % % % % % % % % % % % % % % %
%%%%%%%%%%%%%%%%%%%%%%%%%%%%%%%%%%%%%%%%%%%%%%%%%%%%%%%%%%%%%%%%%%%
\section{Achievable Secure DoF via Opportunistic Jammer Selection}
%%%%%%%%%%%%%%%%%%%%%%%%%%%%%%%%%%%%%%%%%%%%%%%%%%%%%%%%%%%%%%%%%%%
% % % % % % % % % % % % % % % % % % % % % % % % % % % % % % % % % %
%%%%%%%%%%%%%%%%%%%%%%%%%%%%%%%%%%%%%%%%%%%%%%%%%%%%%%%%%%%%%%%%%%%
\subsection{The Concept of Opportunistic Jammer Selection}

The purpose of the opportunistic jammer selection is to obtain the
secure DoF between Alice and Bob.
To hinder Eve's eavesdropping, Bob selects $K$ jammers among $S$
jammers in the jammer group. Since the jamming signals also
interfere with Bob, appropriate jammers should be selected from the
jammer group.
Using the IA concept, the subspace spanned by multiple
$N_j$-dimensional signals can be reduced minimally in the
$N_j$-dimensional subspace. It was also shown in \cite{LC2011_3,
LC2013, LC2013_2} that interference alignment can be achieved by
opportunistic user selection if the number of users goes to
infinity.

Since each jammer has $N_j$ transmit antennas, each jamming signal
spans $N_j$-dimensional subspace in $\mathbb{C}^{N_r}$ at Bob's
receiver. Thus, opportunistically selected jamming signals can be
aligned in $N_j$-dimensional subspace if the selection pool size $S$
goes to infinity.
If the jamming signals are aligned in $(N_r-N_t)$-dimensional
subspace at Bob's receiver, Bob can use the residual $N_t$
dimensions for Alice's signals. The concept of the opportunistic
jammer selection is illustrated in Fig. \ref{fig:OJS_concept}. In
the jammer group, Bob selects $K$ jammers whose channels are most
aligned. When there are an infinite number of jammers (i.e.,
$S=\infty$), the jamming signals can be perfectly aligned in
$N_j$-dimensional subspace at Bob's receiver by proper jammer
selection. In this case, $N_t+N_j$ antennas are enough for Bob to
achieve DoF of $N_t$.
At Eve's receiver, on the other hand, the jamming signals are not
aligned and span $K N_j$-dimensional subspace. This is because the
jammer selection of Bob is independent of Eve so that it corresponds
to a random jammer selection to Eve. Since $KN_j$-dimensional
subspace at Eve's receiver is corrupted by the jamming signals from
the selected $K$ jammers, $KN_j$ can be interpreted as
\emph{defensible dimensions} of the security system against
eavesdropping. If the number of Eve's receive antenna is less than
$KN_j$, Eve cannot achieve any DoF; Eve needs $N_t+KN_j$ receive
antennas for DoF of $N_t$.
Thus, the jamming system is designed according to the target
defensible dimensions.
Since Eve has $N_e$ antennas, we require the defensible DoF larger
than or equal to $N_e$, i.e., $K N_j\ge N_e$ to yield the zero DoF
at Eve.

%%%%%%%%%%%%%%%%%%%%%%%%%%%%%%%%%%%%%%%%%%%%%%%%%%%%%%%%%%%%%%%%%%%
% % % % % % % % % % % % % % % % % % % % % % % % % % % % % % % % % %
%%%%%%%%%%%%%%%%%%%%%%%%%%%%%%%%%%%%%%%%%%%%%%%%%%%%%%%%%%%%%%%%%%%
\subsection{Secure DoF via Opportunistic Jammer Selection}
\label{section:secure_DoF}

In this section, we find the achievable secure DoF via OJS.
We recall the channel capacity of Eve given in \eqref{eqn:R_Eve}:
\begin{align}
 \REve \triangleq \log_2 \left\vert
 \mathbf{I}_{N_e} +
     \frac{P}{N_t}\mathbf{G}_0 \mathbf{G}_0 ^\dagger
    \left(\mathbf{I}_{N_e} + \frac{P}{N_j}
        \sum_{k=1}^K \mathbf{G}_{s_k}\mathbf{G}_{s_k} ^\dagger
    \right)^{-1} \right\vert\NN.
\end{align}
Then, it can be decomposed into $ \REve^{+}$ and $\REve^{-}$ such
that $\REve= \REve^{+} - \REve^{-}$ given by
\begin{align}
 \REve^+&\triangleq
    \log_2 \left\vert
    \mathbf{I}_{N_e} +
     \frac{P}{N_t} \mathbf{G}_0 \mathbf{G}_0 ^\dagger
    + \frac{P}{N_j}\sum_{k=1}^K \mathbf{G}_{s_k} \mathbf{G}_{s_k} ^\dagger
    \right\vert \NNL%\label{eqn:REve+}\\
 \REve^-&\triangleq
    \log_2 \left\vert
    \mathbf{I}_{N_e} + \frac{P}{N_j} \sum_{k=1}^K
    \mathbf{G}_{s_k} \mathbf{G}_{s_k} ^\dagger
    \right\vert.\NN
%    \label{eqn:REve-}
\end{align}
In the above equations, the matrix $\sum_{k=1}^K
\mathbf{G}_{s_k}\mathbf{G}_{s_k} ^\dagger ( = [\mathbf{G}_{s_1},
\ldots, \mathbf{G}_{s_K}][\mathbf{G}_{s_1}, \ldots,
\mathbf{G}_{s_K}]^\dagger)$ becomes an $N_e\times N_e$ Wishart
matrix with $KN_j ~(\ge N_e)$ DoF, and hence it has $N_e$ non-zero
eigenvalues with probability one.
From this fact, we can easily show that Eve's DoF becomes zero for
each channel realization such that
\begin{align}
 \DoF{\REve} = \DoF{\REve^+ - \REve^-}
 =N_e-N_e=0.
 \label{eqn:DoF_Eve}
\end{align}
This implicates that the achievable DoF at Bob directly becomes the
secure DoF.

To find the secure DoF at Bob, we decompose the average achievable
rate at Bob, i.e., $\Avr{\RBob}$, into $\RBob^+$ and $\RBob^{-}$
such that $\Avr{\RBob} = \RBob^{+} - \RBob^{-}$, which are given,
respectively, by
\begin{align}
 \RBob^+&=
    \mathbb{E}\log_2 \left\vert
    \mathbf{I}_{N_t} + \mathbf{V}^\dagger
    \left( \frac{P}{N_t}\mathbf{H}_0 \mathbf{H}_0 ^\dagger
    + \frac{P}{N_j}\sum_{k=1}^K \mathbf{H}_{s_k} \mathbf{H}_{s_k} ^\dagger
    \right)\mathbf{V}
    \right\vert \NNL%\label{eqn:Rbob_gain}\\
 \RBob^-&=
    \mathbb{E}\log_2 \left\vert
    \mathbf{I}_{N_t} + \frac{P}{N_j}
    \sum_{k=1}^K \mathbf{V}^\dagger \mathbf{H}_{s_k} \mathbf{H}_{s_k} ^\dagger \mathbf{V}
    \right\vert.\NN%\label{eqn:Rbob_loss}
\end{align}
Then, similar to \eqref{eqn:DoF_Eve}, the achievable DoF of Bob
becomes
\begin{align}
 \mathbb{E}\left\{\DoF{\RBob}\right\}
 &= \DoF{\RBob^+ - \RBob^-} \NNL
 &= N_t - \DoF{\RBob^-}.
 \label{eqn:DoF_Bob}
\end{align}

In \eqref{eqn:DoF_Eve}, we show that Eve's DoF becomes zero for each
channel realization. By plugging \eqref{eqn:DoF_Eve} and
\eqref{eqn:DoF_Bob} in \eqref{eqn:SDoF}, we obtain the achievable
secure DoF via opportunistic jammer selection:
\begin{align}
 \boxed{\textrm{Secure DoF}= N_t - \DoF{\RBob^-}  .}
 \label{eqn:DoF_secure}
\end{align}
Therefore, the achievable secure DoF depends on how much the DoF
loss from jamming signals is reduced at Bob.
In the latter part of the paper, we set the target secure DoF $d\in
(0, N_t]$ and find the number of required jammers to obtain the
target secure DoF.

%%%%%%%%%%%%%%%%%%%%%%%%%%%%%%%%%%%%%%%%%%%%%%%%%%%%%%%%%%%%%%%%%%%
% % % % % % % % % % % % % % % % % % % % % % % % % % % % % % % % % %
%%%%%%%%%%%%%%%%%%%%%%%%%%%%%%%%%%%%%%%%%%%%%%%%%%%%%%%%%%%%%%%%%%%
\section{Geometric Interpretations}

%%%%%%%%%%%%%%%%%%%%%%%%%%%%%%%%%%%%%%%%%%%%%%%%%%%%%%%%%%%%%%%%%%%
% % % % % % % % % % % % % % % % % % % % % % % % % % % % % % % % % %
%%%%%%%%%%%%%%%%%%%%%%%%%%%%%%%%%%%%%%%%%%%%%%%%%%%%%%%%%%%%%%%%%%%
\subsection{Geometric Interpretations of Jamming Channels}
The Grassmann manifold $\Grassmann{N_r, N_j}$ is set of all
$N_j$-dimensional subspaces in $\mathbb{C}^{N_r}$
\cite{DLR2008, BN2002, H2005}.
Since each jammer and Bob have $N_j$ and $N_r$ antennas,
respectively, the channel matrix from each jammer to Bob constructs
an $N_j$-dimensional subspace in $\mathbb{C}^{N_r}$. Let
$\mathsf{H}_k$ be the subspace formed by the channel matrix from the
$k$th jammer to Bob, i.e., $\mathbf{H}_k$. Then, it belongs to the
Grassmann manifold $\Grassmann{N_r, N_j}$, i.e., $\mathsf{H}_k \in
\Grassmann{N_r, N_j}$.
Each subspace can be represented by the \emph{generator matrix}
whose columns are orthonormal and span the same subspace.
If we denote the generator matrix of $\mathsf{H}_k$ by
$\tilde{\mathbf{H}}_k\in\mathbb{C}^{N_r\times N_j}$, it is satisfied
that $\tilde{\mathbf{H}} _k^\dagger \tilde{\mathbf{H}}_k =
\mathbf{I}_{N_j}$ and $\Span{\tilde{\mathbf{H}}_k} =
\Span{\mathbf{H}_k} = \mathsf{H}_k$.
Since Bob has more antennas than Alice, Bob can partly suppress the
jamming signals using the residual antenna dimensions.
In our case, the $N_j$-dimensional jamming signals should be aligned
in $(N_r - N_t)$-dimensional subspace to obtain the secure DoF of
$N_t$.

The distance between two subspaces can be defined in many ways. The
\emph{chordal distance} is one of the most widely used ones.
Let $\mathsf{Q}\in\Grassmann{N_r, N_r - N_t}$ be an arbitrary $(N_r
- N_t)$-dimensional subspace and $\mathbf{Q} \in \mathbb{C}^{N_r
\times (N_r - N_t)}$ be its generator matrix such that
$\mathbf{Q}^\dagger\mathbf{Q} = \mathbf{I}_{N_r-N_t}$.
Then, the squared chordal distance between $\mathsf{H}_k \in
\Grassmann{N_r, N_j}$ and $\mathsf{Q} \in \Grassmann{N_r, N_r-N_t}$
denoted by $\Chords{\mathsf{H}_k, \mathsf{Q}}$ is calculated from
the generator matrices of them such that
\begin{align}
 \Chords{\mathsf{H}_k, \mathsf{Q}}
 &\triangleq \min(N_j, N_r-N_t)
    -tr(\tilde{\mathbf{H}}_k \tilde{\mathbf{H}}_k^\dagger
    \mathbf{Q}\mathbf{Q}^\dagger)\NNL
 &\stackrel{(a)}{=} N_j -tr(\tilde{\mathbf{H}}_k^\dagger
    \mathbf{Q}\mathbf{Q}^\dagger \tilde{\mathbf{H}}_k  ),
    \label{eqn:CD1}
\end{align}
where the equality $(a)$ holds because $N_r \ge N_t + N_j$ and
$tr(\mathbf{AB}) = tr(\mathbf{BA})$. Note that $\tilde{\mathbf{H}}_k
\tilde{\mathbf{H}}_k^\dagger$ and $\mathbf{Q} \mathbf{Q}^\dagger$
are the projection matrices onto the subspaces $\mathsf{H}_k$ and
$\mathsf{Q}$, respectively. See \cite{BB2007} for more details on
the chordal distance.

%%%%%%%%%%%%%%%%%%%%%%%%%%%%%%%%%%%%%%%%%%%%%%%%%%%%%%%%%%%%%%%%%%%
% % % % % % % % % % % % % % % % % % % % % % % % % % % % % % % % % %
%%%%%%%%%%%%%%%%%%%%%%%%%%%%%%%%%%%%%%%%%%%%%%%%%%%%%%%%%%%%%%%%%%%
\begin{lemma}\label{lemma:CD2}
The squared chordal distance between $\mathsf{H}_k$ and $\mathsf{Q}$
is equivalent to
\begin{align}
 \Chords{\mathsf{H}_k, \mathsf{Q}}
 &= tr((\mathbf{Q}^\perp)^\dagger\tilde{\mathbf{H}}_k
    \tilde{\mathbf{H}}_k^\dagger(\mathbf{Q}^\perp)),
    \label{eqn:CD2}
\end{align}
where $\mathbf{Q}^{\perp}\in\mathbb{C}^{N_r\times N_t}$ is the
generator matrix of the orthogonal complement subspace of
$\mathsf{Q}$ denoted by  $\mathsf{Q}^{\perp} \in \Grassmann{N_r,
N_t}$.
\end{lemma}
%%%%%%%%%%%%%%%%%%%%%%%%%%%%%%%%%%%%%%%%%%%%%%%%%%%%%%%%%%%%%%%%%%%
% % % % % % % % % % % % % % % % % % % % % % % % % % % % % % % % % %
%%%%%%%%%%%%%%%%%%%%%%%%%%%%%%%%%%%%%%%%%%%%%%%%%%%%%%%%%%%%%%%%%%%
\begin{proof}
To prove the equivalence between \eqref{eqn:CD1} and
\eqref{eqn:CD2}, it is enough to show $ tr(\mathbf{Q}^\dagger
\tilde{\mathbf{H}}_k \tilde{\mathbf{H}}_k^\dagger\mathbf{Q}) +
tr((\mathbf{Q}^\perp)^\dagger \tilde{\mathbf{H}}_k
\tilde{\mathbf{H}}_k^\dagger (\mathbf{Q}^\perp)) = N_j$.
Since the concatenated matrix $[\mathbf{Q}, \mathbf{Q}^\perp]$ is an
unitary matrix such that
$$[\mathbf{Q}, \mathbf{Q}^\perp]
[\mathbf{Q}, \mathbf{Q}^\perp]^\dagger= [\mathbf{Q},
\mathbf{Q}^\perp]^\dagger [\mathbf{Q}, \mathbf{Q}^\perp] =
\mathbf{I}_{N_r},$$
it is satisfied that
\begin{align}
 &tr(\mathbf{Q}^\dagger \tilde{\mathbf{H}}_k
    \tilde{\mathbf{H}}_k^\dagger\mathbf{Q})
    +tr((\mathbf{Q}^\perp)^\dagger \tilde{\mathbf{H}}_k
    \tilde{\mathbf{H}}_k^\dagger (\mathbf{Q}^\perp))\NNL
 &\qquad=tr([\mathbf{Q}, \mathbf{Q}^\perp]^\dagger
    \tilde{\mathbf{H}}_k \tilde{\mathbf{H}}_k^\dagger
    [\mathbf{Q}, \mathbf{Q}^\perp])%\NNL
 \stackrel{(a)}{=}tr(\tilde{\mathbf{H}}_k^\dagger [\mathbf{Q}, \mathbf{Q}^\perp]
    [\mathbf{Q}, \mathbf{Q}^\perp]^\dagger \tilde{\mathbf{H}}_k)
 = N_j,\NN
\end{align}
where the equality $(a)$ holds from $tr(\mathbf{A}\mathbf{B}) =
tr(\mathbf{B}\mathbf{A})$.
\end{proof}

%%%%%%%%%%%%%%%%%%%%%%%%%%%%%%%%%%%%%%%%%%%%%%%%%%%%%%%%%%%%%%%%%%%
% % % % % % % % % % % % % % % % % % % % % % % % % % % % % % % % % %
%%%%%%%%%%%%%%%%%%%%%%%%%%%%%%%%%%%%%%%%%%%%%%%%%%%%%%%%%%%%%%%%%%%
\begin{lemma}\label{lemma:unique_subspace}
Any full rank precoding matrix at each jammer cannot change the
jamming subspace at Bob's receiver. For example, when the $k$th
jammer uses an arbitrary precoding matrix $\mathbf{U} \in\mathbb{C}
^{N_j\times N_j}$ of rank $N_j$, it is satisfied that
$\Span{\mathbf{H}_k} = \Span{\mathbf{H}_k \mathbf{U}}=\mathsf{H}_k$.
%
%\begin{align}
% \Span{\mathbf{H}_k}=\Span{\mathbf{H}_k\mathbf{U}}=\mathsf{H}_k.
%\end{align}
%
\end{lemma}
%%%%%%%%%%%%%%%%%%%%%%%%%%%%%%%%%%%%%%%%%%%%%%%%%%%%%%%%%%%%%%%%%%%
\begin{proof}
Since both $\mathbf{H}_k$ and $\mathbf{U}$ are the matrices of rank
$N_j$, so is $\mathbf{H}_k\mathbf{U}$.
Intuitively, the columns of $\mathbf{H}_k \mathbf{U}$ are linear
combinations of the columns of $\mathbf{H}_k$ so that they will span
the same subspace. Formally, we can show $\Chords{\mathbf{H}_k,
\mathbf{H}_k\mathbf{U}}=0$, but it is trivial.
\end{proof}
%%%%%%%%%%%%%%%%%%%%%%%%%%%%%%%%%%%%%%%%%%%%%%%%%%%%%%%%%%%%%%%%%%%
% % % % % % % % % % % % % % % % % % % % % % % % % % % % % % % % % %
%%%%%%%%%%%%%%%%%%%%%%%%%%%%%%%%%%%%%%%%%%%%%%%%%%%%%%%%%%%%%%%%%%%

As stated in Lemma \ref{lemma:unique_subspace}, each jamming channel
forms an unique subspace invariant to any full rank precoding
matrix.
Since there are total $S$ jammers, we have $S$ jamming subspaces
$\mathsf{H}_1, \ldots, \mathsf{H}_S$ such that $\{\mathsf{H}_1,
\ldots, \mathsf{H}_S\} \subset \Grassmann{N_r,N_j}$. Note that each
jamming subspace is isotropic in $\mathbb{C}^{N_r}$ and independent
of each other because all jamming channels have the elements of
i.i.d. complex Gaussian random variables.
A graphical interpretation is given in Fig.
\ref{fig:jamming_subspaces}.
In Fig. \ref{fig:jamming_subspaces}, each point on the sphere
surface is an $N_j$-dimensional subspace represented by its
generator matrix, while the surface of the sphere is an
$N_r$-dimensional space, i.e., $\mathbb{C}^{N_r}$, which is set of
all $N_j$-dimensional subspaces.
Since we should select $K$ jammers among $S$ jammers, we select $K$
points out of $S$ points (jamming subspaces) on the sphere.

%%%%%%%%%%%%%%%%%%%%%%%%%%%%%%%%%%%%%%%%%%%%%%%%%%%%%%%%%%%%%%%%%%%
% % % % % % % % % % % % % % % % % % % % % % % % % % % % % % % % % %
%%%%%%%%%%%%%%%%%%%%%%%%%%%%%%%%%%%%%%%%%%%%%%%%%%%%%%%%%%%%%%%%%%%
\subsection{Alignment Measure among $K$ Jamming Subspaces}

Now, our aim is to select $K$ jammers whose jamming subspaces are
most aligned in an $(N_r-N_t)$-dimensional subspace.
For jammer selection, we define an \emph{alignment measure} to
quantify how well $K$ jamming subspaces are aligned and find the
relationship between the alignment measure and jammer group size
(i.e., $S$). The alignment measure is based on the mini-max distance
of the selected jamming subspaces from an $(N_r-N_t)$-dimensional
subspace.
For example, the alignment measure of the $K$ subspaces formed by
$K$ jammers indexed from $1$ to $K$, i.e., $\mathsf{H}_1, \ldots,
\mathsf{H}_K$, is defined by
\begin{align}
 \mathfrak{q}(\mathsf{H}_1, \ldots, \mathsf{H}_K)
 \triangleq \min_{\mathsf{Q}\in \Grassmann{N_r,N_r-N_t}}
 \max_{k\in[K]}~\Chord{\mathsf{H}_k, \mathsf{Q}},
 \label{eqn:alignment_measure}
\end{align}
where $[K]=\{1, \ldots, K\}$.
If $\mathfrak{q}(\mathsf{H}_1, \ldots, \mathsf{H}_K)=0$, there
exists $\mathsf{Q}\in \Grassmann{N_r,N_r-N_t}$ such that
\begin{align}
 \Chord{\mathsf{H}_1, \mathsf{Q}} = \Chord{\mathsf{H}_2, \mathsf{Q}}
 =\cdots =\Chord{\mathsf{H}_K, \mathsf{Q}}=0,\NN
\end{align}
which means that the $K$ jamming subspaces are perfectly aligned in
an $(N_r-N_t)$-dimensional subspace.
Otherwise, if $\mathfrak{q}(\mathsf{H}_1, \ldots, \mathsf{H}_K) =
\delta$, there exist $\mathsf{Q}\in \Grassmann{N_r, N_r - N_t}$ such
that
\begin{align}
 \Chord{\mathsf{H}_k, \mathsf{Q}}\le \delta \qquad \forall k\in[K],\NN
\end{align}
which means the $K$ subspaces are aligned in an
$(N_r-N_t)$-dimensional subspace within distance $\delta$.
Since there are $S$ jammers in our problem, we can select $K$
jammers (i.e., jamming subspaces) with the smallest alignment
measure among $S$ jammers.
Thus, the alignment measure for the selected jammers becomes
\begin{align}
 \min_{\{s_1', \ldots, s_K'\} \subset [S]}
    \mathfrak{q}(\mathsf{H}_{s_1'}, \ldots, \mathsf{H}_{s_K'}).\NN
\end{align}
%
%
%The selected jammers denoted by $s_1, \ldots, s_K$ are obtained by
%%
%\begin{align}
% (s_1, \ldots, s_K) = \Argmin{\{s_1', \ldots, s_K'\} \subset [S]}
%    \mathfrak{q}(\mathsf{H}_{s_1'}, \ldots, \mathsf{H}_{s_K'}).
%    \label{eqn:smallest_alignment_measure}
%\end{align}
%
As the number of jammers increases, we can select more-aligned
jamming subspaces. Therefore, the question of interest is how small
we can make the alignment measure of the $K$ jamming subspaces by
opportunistic jammer selection for a given number of total jammers
$S$.
To answer this question, we adopt subspace quantization theory
\cite{DLR2008, BN2002, H2005, ACKL2005, MH2008}.

Suppose that we quantize an arbitrary $N_j$-dimensional subspace
into one of $M$ $(N_r-N_t)$-dimensional subspaces. We define a
subspace codebook $\mathcal{Q} \triangleq \{\mathsf{Q}_1, \ldots,
\mathsf{Q}_M\}$ comprised of $M$ $(N_r-N_t)$-dimensional subspaces
such that $\vert\mathcal{Q}\vert=M$ and $\mathcal{Q}\subset
\Grassmann{N_r, N_r-N_t}$.
For the $m$th subspace (or codeword), i.e., $\mathsf{Q}_m$, we
define a metric ball with radius $\delta$:
\begin{align}
 B(\mathsf{Q}_m, \delta)
 \triangleq \{\mathsf{P} \in \Grassmann{N_r, N_j} :
    \Chord{\mathsf{Q}_m,\mathsf{P}} \le \delta\},\NN
\end{align}
as a set of $N_j$-dimensional subspaces within a distance $\delta$
from the $(N_r-N_j)$-dimensional subspace $\mathsf{Q}_m$.

Generally, the performance of a codebook is measured by two
important parameters -- the \emph{packing radius} and the
\emph{covering radius}.
The packing radius of $\mathcal{Q}$ denoted by
$\delta_p(\mathcal{Q})$ is the maximum radius of each metric ball
which is non-overlapped such that \cite{DLR2008, BN2002, H2005}
\begin{align}
 \delta_p(\mathcal{Q}) \triangleq \max \{\delta :
    B(\mathsf{Q}_i, \delta) \cap B(\mathsf{Q}_j, \delta)
    =\emptyset \quad \forall   i\ne j \}.\NN
\end{align}
The covering radius of $\mathcal{Q}$ denoted by $\delta_c
(\mathcal{Q})$ is the minimum radius of the metric ball covering
whole space such that \cite{ACKL2005, MH2008}
\begin{align}
 \delta_c(\mathcal{Q}) \triangleq \min \{\delta :
    B(\mathsf{Q}_1, \delta)\cup\cdots \cup B(\mathsf{Q}_M, \delta)
    = \mathbb{C}^{N_r}\}.
    \label{eqn:covering_radius}
\end{align}
A graphical representation of the packing and the covering radii of
a codebook are given in Fig. \ref{fig:packing_radius} and Fig.
\ref{fig:covering_radius}, respectively.
Obviously, the covering radius is always larger than or equal to the
packing radius, i.e., $\delta_p(\mathcal{Q}) \le
\delta_c(\mathcal{Q})$. Since the union of $M$ metric balls with the
covering radius $\delta_c(\mathcal{Q})$ is $\mathbb{C}^{N_r}$, any
jamming subspace will be contained at least one of the $M$ metric
balls with the covering radius. This fact leads to Remark
\ref{remark:pigeion_hole}.
%
%%%%%%%%%%%%%%%%%%%%%%%%%%%%%%%%%%%%%%%%%%%%%%%%%%%%%%%%%%%%%%%%%%%
% % % % % % % % % % % % % % % % % % % % % % % % % % % % % % % % % %
%%%%%%%%%%%%%%%%%%%%%%%%%%%%%%%%%%%%%%%%%%%%%%%%%%%%%%%%%%%%%%%%%%%
\begin{remark}[Pigeon hole principle]\label{remark:pigeion_hole}
When there are $S=(K-1)M+1$ jamming subspaces, at least $K$
subspaces are contained in the same metric ball among
$B(\mathsf{Q}_1, \delta_c(\mathcal{Q})), \ldots, B(\mathsf{Q}_M,
\delta_c(\mathcal{Q}))$.
The illustration is given in Fig. \ref{fig:jammer_covering}.
\end{remark}
%
%%%%%%%%%%%%%%%%%%%%%%%%%%%%%%%%%%%%%%%%%%%%%%%%%%%%%%%%%%%%%%%%%%%
% % % % % % % % % % % % % % % % % % % % % % % % % % % % % % % % % %
%%%%%%%%%%%%%%%%%%%%%%%%%%%%%%%%%%%%%%%%%%%%%%%%%%%%%%%%%%%%%%%%%%%

Now, we consider two optimal codebooks of size $M$; one maximizes
the packing radius, and the other minimizes the covering radius.
Although the optimal codebooks are not unique, the maximum packing
radius and the minimum covering radius will be uniquely determined
for given codebook size.
With a slight abuse of notation, we denote the maximum packing
radius and the minimum covering radius obtained from the codebooks
of size $M$ by $\delta_p^\star(M)$ and $\delta_c^\star(M)$,
respectively, such that
\begin{align}
 \delta_p^\star(M)
 &=\max_{\mathcal{Q} \subset \Grassmann{N_r, N_r-N_t} \atop \vert\mathcal{Q}\vert=M}
    \delta_p(\mathcal{Q})\NNL
 \delta_c^\star(M)
 &=\min_{\mathcal{Q} \subset \Grassmann{N_r, N_r-N_t} \atop \vert\mathcal{Q}\vert=M}
    \delta_c(\mathcal{Q}).\NN
\end{align}
%

%%%%%%%%%%%%%%%%%%%%%%%%%%%%%%%%%%%%%%%%%%%%%%%%%%%%%%%%%%%%%%%%%%%
% % % % % % % % % % % % % % % % % % % % % % % % % % % % % % % % % %
%%%%%%%%%%%%%%%%%%%%%%%%%%%%%%%%%%%%%%%%%%%%%%%%%%%%%%%%%%%%%%%%%%%

Unfortunately, however, the exact values of the optimal packing and
covering radii are unknown because finding them are NP-hard
problems.
Instead, the optimal packing radius $\delta_p^\star(M)$ is shown to
be \cite{DLR2008}
\begin{align}
 \kappa_1 M^{-\frac{1}{N_t N_j}} (1+o(1))
 \le \delta_p^\star(M)
 \le \kappa_2 M^{-\frac{1}{N_t N_j}} (1+o(1)),
 \label{eqn:packing_result}
\end{align}
where $\kappa_1$ and $\kappa_2$ are constants invariant to $M$ (see
\cite{DLR2008} for details).\footnote{$f(n) = o(g(n))$ means that
$f(n)\le cg(n)$ for any $c>0$ when $n$ is sufficiently large.}
It is also shown that the main terms of both upper and lower bounds
are quite accurate estimates of the optimal packing radius when $M$
is sufficiently large.

Obviously, the optimal covering radius is larger than the optimal
packing radius. It can be easily proved by contradiction that the
optimal covering radius is smaller than the twice of the optimal
packing radius.
So we can establish the following relationship:
\begin{align}
 \delta_p^\star(M) \le \delta_c^\star(M) \le 2\delta_p^\star(M).
 \label{eqn:packing_covering}
\end{align}
From \eqref{eqn:packing_result} and \eqref{eqn:packing_covering}, we
obtain the range of the optimal covering radius in Remark
\ref{remark:covering_bound}.

%%%%%%%%%%%%%%%%%%%%%%%%%%%%%%%%%%%%%%%%%%%%%%%%%%%%%%%%%%%%%%%%%%%
% % % % % % % % % % % % % % % % % % % % % % % % % % % % % % % % % %
%%%%%%%%%%%%%%%%%%%%%%%%%%%%%%%%%%%%%%%%%%%%%%%%%%%%%%%%%%%%%%%%%%%
\begin{remark}\label{remark:covering_bound}
Using codebook of size $M$, the optimal covering radius satisfies
that
\begin{align}
 \kappa_1 M^{-\frac{1}{N_t N_j}} (1+o(1))
 \le \delta_c^\star(M)
 \le 2\kappa_2 M^{-\frac{1}{N_t N_j}} (1+o(1)).\NN
% \label{eqn:covering_range}
\end{align}
When the codebook size is sufficiently large (i.e., when $M$ is
sufficiently large), the optimal covering radius is scaled by the
codebook size $M$ such that
\begin{align}
 \delta_c^\star(M) \propto M^{-\frac{1}{N_t N_j}}.\NN
\end{align}
\end{remark}
%%%%%%%%%%%%%%%%%%%%%%%%%%%%%%%%%%%%%%%%%%%%%%%%%%%%%%%%%%%%%%%%%%%
% % % % % % % % % % % % % % % % % % % % % % % % % % % % % % % % % %
%%%%%%%%%%%%%%%%%%%%%%%%%%%%%%%%%%%%%%%%%%%%%%%%%%%%%%%%%%%%%%%%%%%

The optimal covering radius is closely related to our problem. As
stated in Remark \ref{remark:pigeion_hole}, if there are
$S=(K-1)M+1$ jammers, we can ensure that there exist $K$ jamming
channels aligned in an $(N_r-N_t)$-dimensional subspace within
distance $\delta_c^\star(M)$.
%
%Thus, we prove Lemma \ref{lemma:AM_bound} as follows.
%
Therefore, we establish the following key lemma.

%%%%%%%%%%%%%%%%%%%%%%%%%%%%%%%%%%%%%%%%%%%%%%%%%%%%%%%%%%%%%%%%%%%
% % % % % % % % % % % % % % % % % % % % % % % % % % % % % % % % % %
%%%%%%%%%%%%%%%%%%%%%%%%%%%%%%%%%%%%%%%%%%%%%%%%%%%%%%%%%%%%%%%%%%%
\begin{lemma}\label{lemma:AM_bound}
When there are $S=(K-1)M+1$ jammers, we can always pick $K$ jammers
whose alignment measure is smaller than $2\kappa_2 M^{-\frac{1}{N_t
N_j}} (1+o(1))$, i.e.,
\begin{align}
 \min_{\{s_1',\ldots, s_K'\}\subset [S]}
    \mathfrak{q}(\mathsf{H}_{s_1'}, \ldots, \mathsf{H}_{s_K'})
%    &\le \delta_c^\star(M)\NNL
    &\le 2\kappa_2 M^{-\frac{1}{N_t N_j}} (1+o(1)),\NN
\end{align}
where $\kappa_2$ are constants invariant to $M$, and the term
$(1+o(1))$ can be ignored for large $M$.
\end{lemma}
%%%%%%%%%%%%%%%%%%%%%%%%%%%%%%%%%%%%%%%%%%%%%%%%%%%%%%%%%%%%%%%%%%%
\begin{proof}%See Appendix \ref{appendix:AM_bound}.
Let $\mathcal{Q}_c^\star=\{\mathsf{Q}_1^\star, \ldots,
\mathsf{Q}_M^\star\}$ be the codebook that minimizes the covering
radius, and consider $M$ metric balls $B(\mathsf{Q}_1^\star,
\delta_c^\star(M)), \ldots, B(\mathsf{Q}_M^\star,
\delta_c^\star(M))$. By the definition, the union of the metric
balls becomes $\mathbb{C}^{N_r}$.
%
%it is satisfied that $B(\mathsf{Q}_1^\star, \delta)\cup\cdots \cup
%B(\mathsf{Q}_M^\star, \delta) = \mathbb{C}^{N_r}$.
%
As stated in Remark \ref{remark:pigeion_hole}, when there are
$S=(K-1)M+1$ jammers, we can always pick $K$ jamming subspaces in
the same metric ball.
Let $\mathsf{H}_{(1)},\ldots, \mathsf{H}_{(K)}$ be the $K$ jamming
subspaces in the same metric ball, and $\mathsf{Q}_m^\star$ be the
center of the metric ball.
Then, we obtain
\begin{align}
\min_{\{s_1',\ldots, s_K'\}\subset [S]}
    \mathfrak{q}(\mathsf{H}_{s_1'}, \ldots, \mathsf{H}_{s_K'})
 &\le \mathfrak{q}(\mathsf{H}_{(1)},\ldots, \mathsf{H}_{(K)})\NNL
 &\stackrel{(a)}{\le} \max_{k\in[K]} \Chord{\mathsf{H}_{(k)}, \mathsf{Q}_m^\star}\NNL
 &\stackrel{(b)}{\le} \delta_c^\star(M)\NNL
 &\stackrel{(c)}{\le} 2\kappa_2 M^{-\frac{1}{N_t N_j}} (1+o(1)),\NN
\end{align}
where the inequality $(a)$ is from the definition of the alignment
measure given in \eqref{eqn:alignment_measure}, and the inequality
$(b)$ holds because $\mathsf{H}_{(k)}\in B(\mathsf{Q}_m^\star,
\delta_c^\star(M))$ for all $k \in [K]$. Also, the inequality $(c)$
holds from Remark \ref{remark:covering_bound}.

\end{proof}

%%%%%%%%%%%%%%%%%%%%%%%%%%%%%%%%%%%%%%%%%%%%%%%%%%%%%%%%%%%%%%%%%%%
% % % % % % % % % % % % % % % % % % % % % % % % % % % % % % % % % %
%%%%%%%%%%%%%%%%%%%%%%%%%%%%%%%%%%%%%%%%%%%%%%%%%%%%%%%%%%%%%%%%%%%
%\subsection{Achievable Alignment Measure from $S$ Jammers}

%%%%%%%%%%%%%%%%%%%%%%%%%%%%%%%%%%%%%%%%%%%%%%%%%%%%%%%%%%%%%%%%%%%
% % % % % % % % % % % % % % % % % % % % % % % % % % % % % % % % % %
%%%%%%%%%%%%%%%%%%%%%%%%%%%%%%%%%%%%%%%%%%%%%%%%%%%%%%%%%%%%%%%%%%%

\section{Opportunistic Jammer Selection for Secure DoF}

We recall the secure DoF of Bob given in \eqref{eqn:DoF_secure}:
\begin{align}
 \textrm{Secure DoF}= N_t - \DoF{\RBob^-}.
 \label{eqn:DoF_Bob2}
\end{align}
We need an appropriate jammer selection scheme to reduce the DoF
loss from the jamming signals.
In this section, we propose two opportunistic jammer selection
schemes.
Firstly, we find the minimum DoF loss jammer selection scheme to
achieve the maximum secure DoF. Then, we propose the subspace-based
jammer selection scheme.

%%%%%%%%%%%%%%%%%%%%%%%%%%%%%%%%%%%%%%%%%%%%%%%%%%%%%%%%%%%%%%%%%%%
% % % % % % % % % % % % % % % % % % % % % % % % % % % % % % % % % %
%%%%%%%%%%%%%%%%%%%%%%%%%%%%%%%%%%%%%%%%%%%%%%%%%%%%%%%%%%%%%%%%%%%
\subsection{Minimum DoF Loss Jammer Selection Scheme (OJS1)}
In the minimum DoF loss jammer selection scheme, Bob directly
minimizes the rate loss from the jamming signals, i.e., $\RBob^-$ in
\eqref{eqn:DoF_Bob2}. Correspondingly, the DoF loss is minimized.
The rate loss from the minimum DoF loss jammer selection at Bob
denoted by $\RBob^{- (1)}$ is given by
\begin{align}
 \RBob^{- (1)}
 &=\mathbb{E}\left[\min_{s_1',\ldots, s_K', \mathbf{V}}
    \log_2 \left\vert \mathbf{I}_{N_t} + \frac{P}{N_j} \mathbf{V}^\dagger
    \bigg(
        \sum_{k=1}^K \mathbf{H}_{s_k'} \mathbf{H}_{s_k'} ^\dagger
    \bigg)\mathbf{V}
    \right\vert\right]\NNL
 &\stackrel{(a)}{=}\mathbb{E}\left[
    \min_{s_1',\ldots, s_K'}\log_2
    \prod_{n=N_r-N_t+1}^{N_r} \bigg[1 + \frac{P}{N_j}
    \lambda_n \bigg(
    \sum_{k=1}^K \mathbf{H}_{s_k'} \mathbf{H}_{s_k'} ^\dagger
    \bigg) \bigg]\right],
    \label{eqn:R_Bob1}
\end{align}
where the equality $(a)$ is obtained using the postprocessing matrix
\begin{align}
 \mathbf{V}=\left[\mathbf{v}_{N_r - N_t + 1}
    (\mathbf{A}), \ldots, \mathbf{v}_{N_r}(\mathbf{A})\right],\NN
\end{align}
where $\mathbf{A}= \sum_{k=1}^K \mathbf{H}_{s_k'} \mathbf{H}_{s_k'}
^\dagger$, and $\lambda_n(\cdot)$ and $\mathbf{v}_n(\cdot)$ are the
$n$th largest eigenvalue and corresponding eigenvector of the
matrix, respectively. Thus, the selected jammers at Bob become
\begin{align}
 (s_1,\ldots, s_K)
 = \underset{s_1',\ldots, s_K'}{\arg\min}
    \prod_{n=N_r-N_t+1}^{N_r} \left[1
    + \frac{P}{N_j} \lambda_n \left(
    \sum_{k=1}^K \mathbf{H}_{s_k'} \mathbf{H}_{s_k'}^\dagger
    \right) \right].\NN
\end{align}

%%%%%%%%%%%%%%%%%%%%%%%%%%%%%%%%%%%%%%%%%%%%%%%%%%%%%%%%%%%%%%%%%%%
% % % % % % % % % % % % % % % % % % % % % % % % % % % % % % % % % %
%%%%%%%%%%%%%%%%%%%%%%%%%%%%%%%%%%%%%%%%%%%%%%%%%%%%%%%%%%%%%%%%%%%
\subsection{Subspace-based Jammer Selection Scheme (OJS2)}
In this subsection, we propose the suboptimal jammer selection
scheme using the jamming subspaces.
First of all, we find an upper bound of the minimum rate loss of
Bob, i.e., \eqref{eqn:R_Bob1}, given by
\begin{align}
 \RBob^{-(1)}
 &= \mathbb{E}_{\tilde{\mathbf{H}}, \mathbf{\Lambda}}
    \left[\min_{s_1',\ldots, s_K', \mathbf{V}}
    \log_2 \left\vert \mathbf{I}_{N_t} + \frac{P}{N_j}
    \sum_{k=1}^K \mathbf{V}^\dagger\tilde{\mathbf{H}}_{s_k'}
        \mathbf{\Lambda}_{s_k'} \tilde{\mathbf{H}}_{s_k'} ^\dagger
        \mathbf{V}
    \right\vert\right]\NNL
 &\stackrel{(a)}{\le}
    \mathbb{E}_{\tilde{\mathbf{H}}}
    \left[\min_{s_1',\ldots, s_K', \mathbf{V}}
    \mathbb{E}_{\mathbf{\Lambda}}
    \log_2 \left\vert \mathbf{I}_{N_t} + \frac{P}{N_j}
    \sum_{k=1}^K \mathbf{V}^\dagger \tilde{\mathbf{H}}_{s_k'}
        \mathbf{\Lambda}_{s_k'} \tilde{\mathbf{H}}_{s_k'} ^\dagger
        \mathbf{V} \right\vert\right]\NNL
 &\stackrel{(b)}{\le}
    \mathbb{E}_{\tilde{\mathbf{H}}}
    \left[\min_{s_1',\ldots, s_K', \mathbf{V}}
    \log_2 \left\vert \mathbf{I}_{N_t} + \frac{P}{N_j}
        \sum_{k=1}^K \mathbf{V}^\dagger \tilde{\mathbf{H}}_{s_k'}
        \mathbb{E}_{\mathbf{\Lambda}}[\mathbf{\Lambda}_{s_k'}]
        \tilde{\mathbf{H}}_{s_k'} ^\dagger \mathbf{V}
    \right\vert\right]\NNL
 &= \mathbb{E}_{\tilde{\mathbf{H}}}
    \left[\min_{s_1',\ldots, s_K', \mathbf{V}}
    \log_2 \left\vert \mathbf{I}_{N_t} + P
        \sum_{k=1}^K \mathbf{V}^\dagger \tilde{\mathbf{H}}_{s_k'}
        \tilde{\mathbf{H}}_{s_k'} ^\dagger\mathbf{V}
    \right\vert\right],\label{eqn:R_Bob_upper1}
\end{align}
where $\tilde{\mathbf{H}}\in\mathbb{C}^{N_r\times N_j}$ and
$\mathbf{\Lambda}\in \mathbb{R}^{N_j\times N_j}$ are obtained from
$\mathbf{H}\in\mathbb{C}^{N_r\times N_j}$ by compact singular value
decomposition of $\mathbf{H}\mathbf{H}^\dagger$ (i.e.,
$\mathbf{H}\mathbf{H}^\dagger = \tilde{\mathbf{H}}\mathbf{\Lambda}
\tilde{\mathbf{H}}^\dagger$) where $\mathbf{\Lambda}$ is a diagonal
matrix whose diagonal elements are the \emph{unordered} non-zero
eigenvalues of $\mathbf{H}\mathbf{H}^\dagger$. Then,
$\tilde{\mathbf{H}}$ becomes the generator matrix of $\mathbf{H}$.
Note that the generator matrix $\tilde{\mathbf{H}}$ for the matrix
$\mathbf{H}$ is not unique, but any generator matrix yields the same
$\tilde{\mathbf{H}} \tilde{\mathbf{H}}^\dagger$.
The inequality $(a)$ is from the fact that the minimum of the
averages is larger than the average of the minimums.
The inequality $(b)$ holds from Jensen's inequality and the
independence between $\tilde{\mathbf{H}}$ and $\mathbf{\Lambda}$.
Finally, we obtain \eqref{eqn:R_Bob_upper1} from the fact that
$\mathbb{E}[\mathbf{\Lambda}_{s_k'}]=N_j\mathbf{I}_{N_j}$
\cite{RJ2008}.

In the subspace-based jammer selection scheme, we minimize the upper
bound given in \eqref{eqn:R_Bob_upper1} denoted by $\RBob^{-(2)}$
instead of \eqref{eqn:R_Bob1}.
We can rewrite $\RBob^{-(2)}$ as
\begin{align}
 \RBob^{-(2)}
 &\triangleq\mathbb{E}_{\tilde{\mathbf{H}}}
    \left[\min_{s_1',\ldots, s_K', \mathbf{V}}
    \log_2 \left\vert \mathbf{I}_{N_t} + P
        \sum_{k=1}^K \mathbf{V}^\dagger \tilde{\mathbf{H}}_{s_k'}
        \tilde{\mathbf{H}}_{s_k'} ^\dagger\mathbf{V}
    \right\vert\right]\NNL
 &= \mathbb{E}_{\tilde{\mathbf{H}}}\left[
    \min_{s_1',\ldots, s_K'}\log_2
    \prod_{n=N_r-N_t+1}^{N_r} \bigg[1 + P
    \lambda_n \bigg(
    \sum_{k=1}^K \tilde{\mathbf{H}}_{s_k'} \tilde{\mathbf{H}}_{s_k'} ^\dagger
    \bigg) \bigg]\right]
    \label{eqn:R_Bob_upper2}
\end{align}
by applying $\mathbf{V}=\left[\mathbf{v}_{N_r-N_t+1} (\mathbf{B}),
\ldots, \mathbf{v}_{N_r}(\mathbf{B})\right]$ in
\eqref{eqn:R_Bob_upper1} where $\mathbf{B}=\sum_{k=1}^K
\tilde{\mathbf{H}}_{s_k'} \tilde{\mathbf{H}}_{s_k'} ^\dagger$.
Thus, the selected jammers are given by
\begin{align}
 (s_1,\ldots, s_K)
 = \underset{s_1',\ldots, s_K'}{\arg\min}
    \prod_{n=N_r-N_t+1}^{N_r} \left[1
    + P \lambda_n \left(
    \sum_{k=1}^K \tilde{\mathbf{H}}_{s_k'} \tilde{\mathbf{H}}_{s_k'}^\dagger
    \right) \right].\NN
\end{align}

OJS1 considers the channel matrix itself so that the jammer
selection criterion involves the channel magnitude. As a result,
OJS1 becomes the secure DoF optimal jammer selection scheme
minimizing the Bob's DoF loss. On the other hand, OJS2 considers the
subspace spanned by the channel matrix so that all channel matrices
spanning the same subspace are considered identical regardless of
the channel magnitude.

%%%%%%%%%%%%%%%%%%%%%%%%%%%%%%%%%%%%%%%%%%%%%%%%%%%%%%%%%%%%%%%%%%%
% % % % % % % % % % % % % % % % % % % % % % % % % % % % % % % % % %
%%%%%%%%%%%%%%%%%%%%%%%%%%%%%%%%%%%%%%%%%%%%%%%%%%%%%%%%%%%%%%%%%%%
\section{Sufficient Jammer Scaling for Target Secure DoF}

In this section, we find a sufficient jammer scaling law for a
target secure DoF.
In previous section, we proposed two jammer selection schemes which
obtain the minimum rate loss $\RBob^{-(1)}$ and its upper bound
$\RBob^{-(2)}$, respectively.
From \eqref{eqn:DoF_secure}, the target secure DoF $d \in (0, N_t]$
is obtained in both schemes when
\begin{align}
 \DoF{\RBob^{-(1)}}=\DoF{\RBob^{-(2)}}=N_t - d.\NN
\end{align}
Since it is hard to directly find the required jammer scaling for
the target secure DoF, we will find $\RBob^{-(3)}$ as a further
upper bound of $\RBob^{-(1)}$ and find the sufficient jammer scaling
law for the target secure DoF $d$. This scaling ensures that
\begin{align}
 \DoF{\RBob^{-(1)}}\le \DoF{\RBob^{-(2)}}\le \DoF{\RBob^{-(3)}}
 = N_t - d.
 \label{eqn:R_Bob123}
\end{align}

The term $\RBob^{-(2)}$ given in \eqref{eqn:R_Bob_upper2} is upper
bounded as follows
\begin{align}
 \RBob^{-(2)}
 &=\mathbb{E}_{\tilde{\mathbf{H}}}
    \left[\min_{s_1',\ldots, s_K', \mathbf{V}}
    \log_2 \left\vert \mathbf{I}_{N_t} + P
        \sum_{k=1}^K \mathbf{V}^\dagger \tilde{\mathbf{H}}_{s_k'}
        \tilde{\mathbf{H}}_{s_k'} ^\dagger\mathbf{V}
    \right\vert\right] \NNL
 &\stackrel{(a)}{\le}\mathbb{E}_{\tilde{\mathbf{H}}}
    \left[\min_{s_1',\ldots, s_K', \mathbf{V}}
    N_t \log_2
    \left(1 + \frac{P}{N_t}
        \sum_{k=1}^K tr\big(\mathbf{V}^\dagger \tilde{\mathbf{H}}_{s_k'}
        \tilde{\mathbf{H}}_{s_k'} ^\dagger\mathbf{V}\big)
    \right)\right] \NNL
 &\stackrel{(b)}{=}\mathbb{E}_{\tilde{\mathbf{H}}}
    \Bigg[\min_{s_1',\ldots, s_K', \atop \mathsf{V}\in\Grassmann{N_r, N_t}}
    N_t \log_2
    \left(1 + \frac{P}{N_t}
        \sum_{k=1}^K
        \Chords{\mathsf{H}_{s_k'}, \mathsf{V}^\perp}
    \right)\Bigg] \NNL
 &\stackrel{(c)}{\le}\mathbb{E}_{\tilde{\mathbf{H}}}
    \Bigg[\min_{s_1',\ldots, s_K'}
    N_t \log_2
    \left(1 + \frac{P}{N_t}
        K \left[\mathfrak{q}(\mathsf{H}_{s_1'}, \ldots, \mathsf{H}_{s_K'})
    \right]^2\right)\Bigg] \NNL
 &\stackrel{(d)}{\le}
    N_t \log_2
    \left(1 + \frac{4\kappa_2^2KP}{N_t}
        \left\lfloor\frac{S-1}{K-1}\right\rfloor^{-\frac{2}{N_t N_j}}
    \right) + o(1), \label{eqn:R_Bob_upper3}
\end{align}
where $\mathsf{V}\in \Grassmann{N_r, N_t}$ is the subspace formed by
the postprocessing matrix $\mathbf{V}$, and $\mathsf{V}^\perp \in
\Grassmann{N_r, N_r - N_t}$ is the orthogonal complement subspace of
$\mathsf{V}$.
In the above equations, the inequality $(a)$ is from Jensen's
inequality such that $\log (\alpha\beta) \le 2\log[(\alpha +
\beta)/2]$, and the equality $(b)$ holds from Lemma \ref{lemma:CD2}.
The inequality $(c)$ holds because selecting the subspace
$\mathsf{V}$ is equivalent to selecting the subspace
$\mathsf{V}^\perp$, and it is satisfied that
\begin{align}
 \min_{\mathsf{V}^\perp \in \Grassmann{N_r,N_r-N_t}}
     \sum_{k=1}^K \Chords{\mathsf{H}_{s_k'}, \mathsf{V}^\perp}
 \le K \left[\mathfrak{q}
    (\mathsf{H}_{s_1'}, \ldots, \mathsf{H}_{s_K'})\right]^2\NN
\end{align}
from the definition of the alignment measure given in
\eqref{eqn:alignment_measure}.
Also, the inequality $(d)$ is from Lemma \ref{lemma:AM_bound} by
substituting $M=\lfloor(S-1)/(K-1)\rfloor$.
We define \eqref{eqn:R_Bob_upper3} as $\RBob^{-(3)}$, i.e.,
\begin{align}
 \RBob^{-(3)} \triangleq  N_t \log_2
 \left(1 + \frac{4\kappa_2^2PK}{N_t}
    \left\lfloor\frac{S-1}{K-1}\right\rfloor^{-\frac{2}{N_t N_j}}
 \right).\NN
\end{align}
%

%%%%%%%%%%%%%%%%%%%%%%%%%%%%%%%%%%%%%%%%%%%%%%%%%%%%%%%%%%%%%%%%%%%
% % % % % % % % % % % % % % % % % % % % % % % % % % % % % % % % % %
%%%%%%%%%%%%%%%%%%%%%%%%%%%%%%%%%%%%%%%%%%%%%%%%%%%%%%%%%%%%%%%%%%%
\begin{theorem} \label{theorem:Delta_loss}
The sufficient number of jammers to ensure Bob's
rate loss smaller than $\Delta$ is given by
\begin{align}
 S = (K-1)\left[\frac{4\kappa_2^2
 KP}{N_t\left(2^{\Delta/N_t}-1\right)}\right]^{\frac{N_tN_j}{2}} + 1.
 \label{eqn:sufficient_jammers}
\end{align}
\end{theorem}
%%%%%%%%%%%%%%%%%%%%%%%%%%%%%%%%%%%%%%%%%%%%%%%%%%%%%%%%%%%%%%%%%%%
% % % % % % % % % % % % % % % % % % % % % % % % % % % % % % % % % %
%%%%%%%%%%%%%%%%%%%%%%%%%%%%%%%%%%%%%%%%%%%%%%%%%%%%%%%%%%%%%%%%%%%
\begin{proof}
If $\RBob^{-(3)}=\Delta$, Bob can have rate loss smaller than
$\Delta$. By solving $\RBob^{-(3)}=\Delta$, we obtain
\eqref{eqn:sufficient_jammers}.
\end{proof}

%%%%%%%%%%%%%%%%%%%%%%%%%%%%%%%%%%%%%%%%%%%%%%%%%%%%%%%%%%%%%%%%%%%
% % % % % % % % % % % % % % % % % % % % % % % % % % % % % % % % % %
%%%%%%%%%%%%%%%%%%%%%%%%%%%%%%%%%%%%%%%%%%%%%%%%%%%%%%%%%%%%%%%%%%%
\begin{theorem}\label{theorem:Scaling}
The target secure DoF of $d\in(0, N_t]$ is achieved when the number
of jammers is scaled as
$S \propto P^{dN_j/2}$.%
%
%%%\footnote{$f(n) = \mathcal{O}(g(n))$ means that there exists a
%%%constant $c>0$ for sufficiently large $n$ such that $f(n)\le
%%%cg(n)$.}
\end{theorem}
%%%%%%%%%%%%%%%%%%%%%%%%%%%%%%%%%%%%%%%%%%%%%%%%%%%%%%%%%%%%%%%%%%%
% % % % % % % % % % % % % % % % % % % % % % % % % % % % % % % % % %
%%%%%%%%%%%%%%%%%%%%%%%%%%%%%%%%%%%%%%%%%%%%%%%%%%%%%%%%%%%%%%%%%%%
\begin{proof}
As stated in \eqref{eqn:R_Bob123}, the target DoF of $d$ is achieved
when $\tDoF{\RBob^{-(3)}}=N_t-d$,
equivalently,
\begin{align}
 \lim_{P\to\infty} N_t \log_2
    \left(1 + \frac{4\kappa_2^2KP}{N_t}
        \left\lfloor\frac{S-1}{K-1}\right\rfloor^{-\frac{2}{N_t N_j}}
    \right) = \log_2 P^{N_t - d}.\NN
\end{align}
This condition holds when $S\propto P^{dN_j/2}$.
\end{proof}
%%%%%%%%%%%%%%%%%%%%%%%%%%%%%%%%%%%%%%%%%%%%%%%%%%%%%%%%%%%%%%%%%%%
% % % % % % % % % % % % % % % % % % % % % % % % % % % % % % % % % %
%%%%%%%%%%%%%%%%%%%%%%%%%%%%%%%%%%%%%%%%%%%%%%%%%%%%%%%%%%%%%%%%%%%

Interestingly, the scaling law of the required jammers in Theorem
\ref{theorem:Scaling} is affected by neither the number of selected
jammers $K$ nor the number of receive antennas $N_r$ (but,
necessarily, $N_t + N_j \le N_r< N_t + KN_j$).

%%%%%%%%%%%%%%%%%%%%%%%%%%%%%%%%%%%%%%%%%%%%%%%%%%%%%%%%%%%%%%%%%%%
% % % % % % % % % % % % % % % % % % % % % % % % % % % % % % % % % %
%%%%%%%%%%%%%%%%%%%%%%%%%%%%%%%%%%%%%%%%%%%%%%%%%%%%%%%%%%%%%%%%%%%
\begin{remark}
Theorem \ref{theorem:Scaling} implies that using single-antenna
jammers is more efficient in terms of the minimum required jammer
scaling than using multi-antenna jammers for target secure DoF. That
is, $S\propto P^{dN_j/2}$ is minimized when $N_j=1$.
\end{remark}
%%%%%%%%%%%%%%%%%%%%%%%%%%%%%%%%%%%%%%%%%%%%%%%%%%%%%%%%%%%%%%%%%%%
% % % % % % % % % % % % % % % % % % % % % % % % % % % % % % % % % %
%%%%%%%%%%%%%%%%%%%%%%%%%%%%%%%%%%%%%%%%%%%%%%%%%%%%%%%%%%%%%%%%%%%

As shown in Theorem \ref{theorem:Delta_loss}, the required number of
jammers increases with the number of the selected jammers to achieve
a given rate loss; if the number of selected jammers increases from
$K$ to $K+1$, the number of required jammers to maintain Bob's rate
loss smaller than $\Delta$ increases from $S$ to
$$\frac{K}{K-1}\left(\frac{K+1}{K}\right)^{\frac{N_tN_j}{2}}S$$
which becomes $S$ when $K$ is sufficiently large.
However, it does not change the scaling law for the given target
secure DoF as shown in Theorem \ref{theorem:Scaling} because the
secure DoF is defined for the asymptotic case that $P\to\infty$.

More receive antennas are likely to be beneficial for jamming signal
alignment because a larger number of receive antennas increase the
subspace dimensions where the jamming signals (or jamming subspaces)
should be aligned, i.e., $N_r-N_t$. However, a larger number of
receive antennas also increase signal space of each jamming signal,
i.e., $N_r$, and thus make it harder to align all jamming signals
together. This counter-effect cancels the benefit of increased
subspace dimensions for jamming signal alignment and makes the
scaling law independent of the number of receive antennas. For
example, suppose that we can increase Bob's receive antennas $N_r$
to $N_r+1$ while maintaining the antenna configuration of $N_t + N_j
\le N_r + 1< N_t + KN_j$.
Although we have increased dimensions $(N_r+1-N_t)$ for jamming
subspace alignment, each jamming subspace becomes $N_j$-dimensional
subspaces in $\mathbb{C}^{N_r+1}$. The alignment of
$N_j$-dimensional subspaces in  $\mathbb{C}^{N_r+1}$ is much more
difficult than the alignment of $N_j$-dimensional subspaces in
$\mathbb{C}^{N_r}$. The scaling law of the required number of
jammers is unchanged due to this difficulty in spite of the
increased dimensions for jamming subspace alignment. These results
and insights are generalized by the following theorem.

%%%%%%%%%%%%%%%%%%%%%%%%%%%%%%%%%%%%%%%%%%%%%%%%%%%%%%%%%%%%%%%%%%%
% % % % % % % % % % % % % % % % % % % % % % % % % % % % % % % % % %
%%%%%%%%%%%%%%%%%%%%%%%%%%%%%%%%%%%%%%%%%%%%%%%%%%%%%%%%%%%%%%%%%%%
\begin{theorem}\label{theorem:Scaling2}
Any wireless communication system can achieve the secure DoF of
$d\in (0,N_t]$ via opportunistic jammer selection with a jammer
scaling law $S \propto P^{d/2}$.
In this case, each jammer should have a single antenna, and
$(N_t+1)$ antennas are enough for Bob's receiver. We can increase
the defensible dimensions of the system as many as we want by
increasing $K$ to make the achievable DoF of Bob directly become the
secure DoF.
\end{theorem}

%%%%%%%%%%%%%%%%%%%%%%%%%%%%%%%%%%%%%%%%%%%%%%%%%%%%%%%%%%%%%%%%%%%
% % % % % % % % % % % % % % % % % % % % % % % % % % % % % % % % % %
%%%%%%%%%%%%%%%%%%%%%%%%%%%%%%%%%%%%%%%%%%%%%%%%%%%%%%%%%%%%%%%%%%%
\subsection{OJS with the Partial CSI}\label{section:partial_CSI}
In this subsection, we show that Theorem \ref{theorem:Scaling} and
Theorem \ref{theorem:Scaling2} are still valid with secrecy outage
probability $\epsilon$ for practical partial CSI scenarios that CSI
for Eve's channel is not known to Alice. Since Eve's CSI is hard to
know even when Eve's presence is known, many schemes that work in
the absence of Eve's CSI have been proposed \cite{GLG2008, ZGAH2011,
DLG2011}. The authors in \cite{DLG2011} proposed relay chatting for
half-duplexing two-hop amplify-and-forward (AF) relay communication
systems. In the relay chatting scheme, the best relay forwards the
received signal and the other relays send jamming signals over the
null space of the desired channel via distributed beamforming at
each stage. Compared to the relay chatting scheme, each jammer in
our proposed scheme simply transmits an i.i.d. Gaussian signal; the
selected jammers do not require any CSI or joint transmission.
Moreover, we consider the defensible dimension to prevent Eve's
eavesdropping in multiple antenna configurations.

In Section \ref{section:secure_DoF}, we showed that the jamming
signals make Eve's DoF zero for each channel realization.
In this case, Eve's channel capacity given in \eqref{eqn:R_Eve} is
saturated to
\begin{align}
 \lim_{P\to\infty} \REve =
 \log_2 \left\vert
  \mathbf{I}_{N_e} +
     \frac{N_j}{N_t}\mathbf{G}_0 \mathbf{G}_0 ^\dagger
    \left(
        \sum_{k=1}^K \mathbf{G}_{s_k}\mathbf{G}_{s_k} ^\dagger
    \right)^{-1} \right\vert \label{eqn:R_Eve_satur}.
\end{align}
Since the selected jammers are random to Eve, the distribution of
$\lim_{P\to\infty} \REve$ will be identical with that of a random
variable $R$ defined by
\begin{align}
 R \triangleq  \log_2 \left\vert
  \mathbf{I}_{N_e} +
     \frac{N_j}{N_t}\mathbf{G}_0 \mathbf{G}_0 ^\dagger
    \left(
        \sum_{k=1}^K \mathbf{G}_{k}\mathbf{G}_{k} ^\dagger
    \right)^{-1} \right\vert \label{eqn:X},
\end{align}
which can be numerically found. We can choose a constant rate $r$ to
yield
\begin{align}
 \PR{R \ge r} = \epsilon,
 \label{eqn:secure_outage_prob}
\end{align}
where $\epsilon\in [0,1]$. Then, using Wyner's encoding scheme with
two rates $(\RBob, r)$ instead of $(\RBob, \REve)$, Bob can achieve
the secure DoF
\begin{align}
% \mathbb{E}\left\{ \DoF{\left[\RBob - \REve\right]^+} \right\}\gtrsim
  \mathbb{E}\left\{ \DoF{\left[\RBob - r\right]^+} \right\}
  =\mathbb{E}\left\{ \DoF{\RBob} \right\}\NN
\end{align}
with secrecy outage probability $\epsilon$, which is the same as the
achievable secure DoF with the knowledge of $\REve$.

Obviously, in \eqref{eqn:secure_outage_prob}, the smaller $\epsilon$
requires the larger $r$, but it is independent of $P$. Therefore, we
can almost surely obtain the same target secure DoF by choosing
sufficiently large $r$ which makes $\epsilon\approx 0$. Note that
this result comes from the independency between the jamming signals
and Eve's CSI.

%%%%%%%%%%%%%%%%%%%%%%%%%%%%%%%%%%%%%%%%%%%%%%%%%%%%%%%%%%%%%%%%%%%
% % % % % % % % % % % % % % % % % % % % % % % % % % % % % % % % % %
%%%%%%%%%%%%%%%%%%%%%%%%%%%%%%%%%%%%%%%%%%%%%%%%%%%%%%%%%%%%%%%%%%%
\section{Numerical Results}

In this section, we evaluate our proposed opportunistic jammer
selection schemes.
Fig. \ref{fig:achievable_rate} shows that Bob can increase the
achievable rate via jammer selection while maintaining the same
average capacity of Eve's channel. The numbers of antennas at Alice,
Bob, jammers, and Eve are assumed by $(N_t, N_j, N_r, N_e) = (2, 2,
4, 4)$, respectively, and Bob selects two jammers  (i.e., $K=2$) in
the jammer group. In this case, the number of Eve's receive antennas
does not exceed the defensible dimensions, i.e., $N_e\le KN_j$, so
that Eve obtains zero DoF. We also consider the capacity maximizing
jammer selection scheme at Bob for comparison. As shown in Fig.
\ref{fig:achievable_rate}, the channel capacities of both Bob and
Eve are saturated in the high SNR region because the number of
jammers is finite.
However, at a fixed SNR, the achievable rate of Bob increases with
the number of candidate jammers, while Eve's channel capacity
remains unchanged.
The gap between the capacities of Bob and Eve increases with the
number of candidate jammers, and results in the secure DoF of $N_t$
when $S\to\infty$. As the number of jammers increases, Bob can
reduce the negative effects of the jamming signals, and the channel
capacity with jammers by the capacity maximizing jammer selection
will go to that without jamming signals. Since Eve's capacity with
jammers is saturated in the high SNR region, Bob can obtain the
secure DoF without Eve's CSI as described in Section
\ref{section:partial_CSI}.
In the case that $(N_t, N_j, N_r, N_e) = (2, 2, 4, 4)$ and $K=2$,
Fig. \ref{fig:outage_probability} shows the secrecy outage
probability (i.e., $\epsilon$) when the Eve's achievable rate is
treated as a constant (i.e., $r$). This figure shows that the
secrecy outage probability decreases if the constant rate of Eve
increases.

The achievable secrecy rates are plotted in Fig. \ref{fig:scaling1}
when Bob selects two jammers in the jammer group with a scaled
number of jammers. The number of antennas at each node is $(N_t,
N_j, N_r, N_e) = (1, 2, 3, 3)$. In this case, the number of
defensible dimensions of the security system is $KN_j=4$ and hence
the Eve's DoF becomes zero. In Fig. \ref{fig:scaling1}, we consider
the optimal jammer selection scheme at Alice with global CSI to
maximize the secrecy rate, while in our proposed jammer selection
scheme Bob selects jammers with only its own CSI.
As a referential upper bound, we also consider the scheme of
\cite[p.189]{BB2011}. In the referential scheme, Alice with
($N_t+KN_j$) antennas steers beams with perfect CSIT and sends
artificial noise along with information-bearing messages without any
help of jammers. On the other hand, Bob and Eve have $N_r$ and $N_e$
antennas, respectively, as in our system model.
As shown in Theorem \ref{theorem:Scaling}, Bob can obtain DoF of one
when the number of jammers is scaled by $S \propto {P^{dN_j/2}} =
P$.
We consider two jammer scalings $S=P$ and $S=0.3P$. In both cases,
Bob obtains DoF of one which directly becomes secure DoF. In Fig.
\ref{fig:scaling2}, we consider two scenarios $S=P$ and $S=P^{0.5}$
in the same configuration. As predicted in Theorem
\ref{theorem:Scaling}, Bob obtains secure DoF of one and a half when
$S=P$ and $S=P^{0.5}$, respectively.
Note that the referential scheme \cite[p.189]{BB2011} obtains much
higher secrecy rate because Alice not only has more transmit
antennas (i.e., $N_t+KN_j$ antennas) but also exploits perfect CSIT.

In Fig. \ref{fig:scaling3}, we consider the antenna configuration
$(N_t, N_j, N_r, N_e) = (2, 1, 3, 2)$ and plot the achievable
secrecy rates when Bob selects two jammers in the jammer group with
a scaled numbers of jammers $S\propto P$ and $S\propto P^{0.5}$,
respectively. As predicted in Theorem 2, Bob obtains secure DoF two
and one when $S=P$ and $S=P^{0.5}$, respectively.

%%%%%%%%%%%%%%%%%%%%%%%%%%%%%%%%%%%%%%%%%%%%%%%%%%%%%%%%%%%%%%%%%%%
% % % % % % % % % % % % % % % % % % % % % % % % % % % % % % % % % %
%%%%%%%%%%%%%%%%%%%%%%%%%%%%%%%%%%%%%%%%%%%%%%%%%%%%%%%%%%%%%%%%%%%
\section{Conclusions}
In this paper, we proposed the opportunistic jammer selection
schemes to achieve the secure DoF in a secure communication system
aided by jammers. For the opportunistic jammer selection, we
proposed two selection criteria -- the minimum DoF loss jammer
selection and the subspace-based jammer selection. We proved that
the secure DoF can be obtained by aligning jamming signals in a
small dimensional subspace at Bob's receiver through the
opportunistic jammer selection.
From the geometric interpretation, we found the required jammer
scaling laws to obtain target secure DoF at Bob's receiver.

%%%
%%%%%%%%%%%%%%%%%%%%%%%%%%%%%%%%%%%%%%%%%%%%%%%%%%%%%%%%%%%%%%%%%%%%%%%%
%%%\appendices
%%%\def\thesection{\Alph{section}}%
%%%\def\thesectiondis{\Alph{section}}%
%%%%%%%%%%%%%%%%%%%%%%%%%%%%%%%%%%%%%%%%%%%%%%%%%%%%%%%%%%%%%%%%%%%%%%%%
%%%
%%%%%%%%%%%%%%%%%%%%%%%%%%%%%%%%%%%%%%%%%%%%%%%%%%%%%%%%%%%%%%%%%%%%%%%%
%%%\section{Proof of Lemma \ref{lemma:AM_bound}}
%%%\setcounter{equation}{0}
%%%\renewcommand{\theequation}{A.\arabic{equation}}
%%%\label{appendix:AM_bound}
%%%%%%%%%%%%%%%%%%%%%%%%%%%%%%%%%%%%%%%%%%%%%%%%%%%%%%%%%%%%%%%%%%%%%%%%
%%%
%%%To prove Lemma \ref{lemma:AM_bound}, we firstly introduce some
%%%geometrical concepts such as the packing radius and the covering
%%%radius.
%%%%
%%%Then, we show that our problem is directly related with the covering
%%%radius, and find the scaling of the optimal covering radius
%%%corresponding to the codebook size. Finally, we prove Lemma
%%%\ref{lemma:AM_bound} using the pigeon hole principle and the scaling
%%%of the covering radius. To answer this question, we adopt subspace
%%%quantization theory.

%%%%%%%%%%%%%%%%%%%%%%%%%%%%%%%%%%%%%%%%%%%%%%%%%%%%%%%%%%%%%%%%%%%
% % % % % % % % % % % % % % % % % % % % % % % % % % % % % % % % % %
%%%%%%%%%%%%%%%%%%%%%%%%%%%%%%%%%%%%%%%%%%%%%%%%%%%%%%%%%%%%%%%%%%%
\clearpage \linespread{1.8777}

%%%%%%%%%%%%%%%%%%%%%%%%%%%%%%%%%%%%%%%%%%%%%%%%%%%%%%%%%%%%%%%%%%%
% % % % % % % % % % % % % % % % % % % % % % % % % % % % % % % % % %
%%%%%%%%%%%%%%%%%%%%%%%%%%%%%%%%%%%%%%%%%%%%%%%%%%%%%%%%%%%%%%%%%%%

%%%%%%%%%%%%%%%%%%%%%%%%%%%%%%%%%%%%%%%%%%%%%%%%%%%%%%%%%%%%%%%%%%%
% % % % % % % % % % % % % % % % % % % % % % % % % % % % % % % % % %
%%%%%%%%%%%%%%%%%%%%%%%%%%%%%%%%%%%%%%%%%%%%%%%%%%%%%%%%%%%%%%%%%%%
\clearpage

\begin{figure}[!t]
\centering
  \includegraphics[width=0.8777\columnwidth]{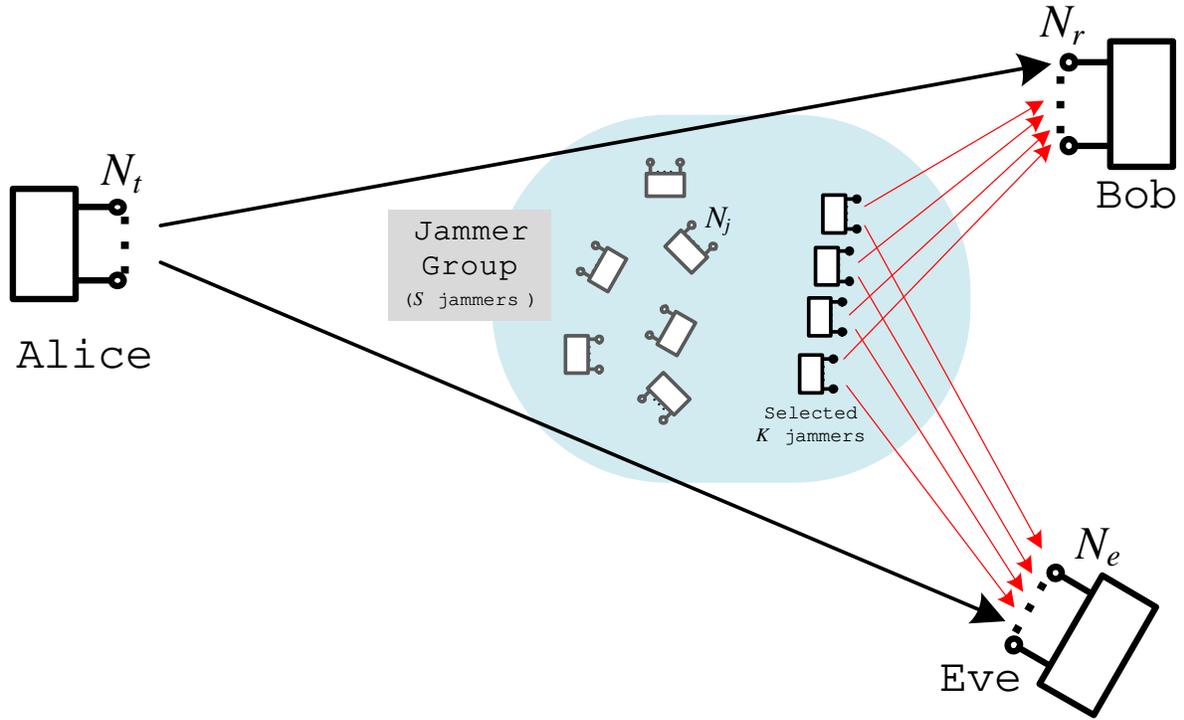}\\
  \caption{System model.}
  \label{fig:system_model}
\end{figure}

\begin{figure}[!t]
\centering
  \includegraphics[width=0.8\columnwidth]{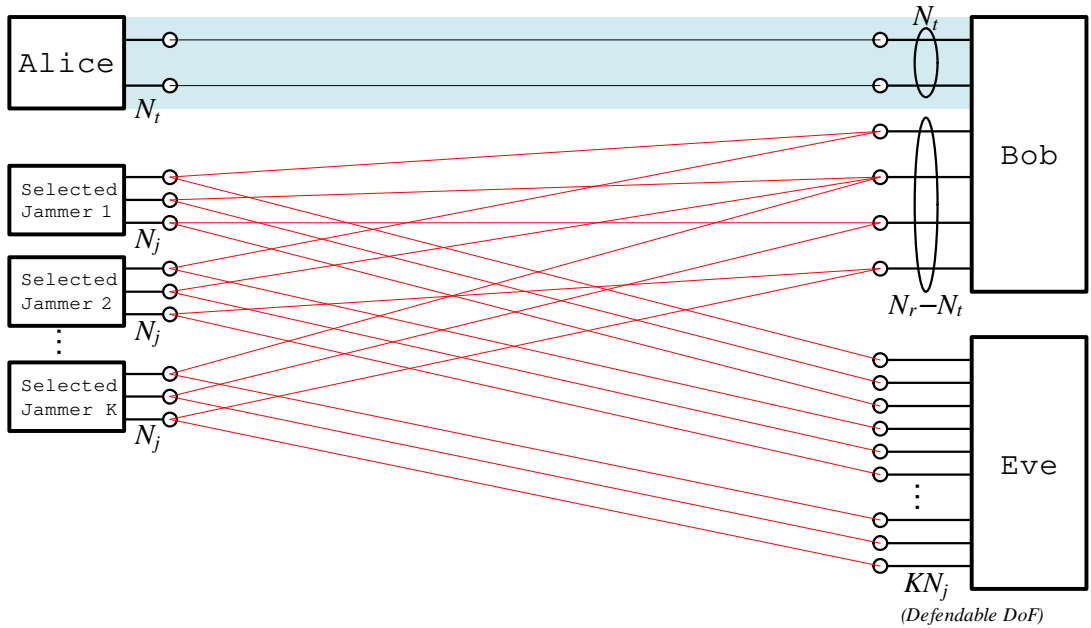}\\
  \caption{The basic concept of opportunistic jammer selection. The
  jamming signals from the selected jammers are aligned at Bob but
  not at Eve.}
  \label{fig:OJS_concept}
\end{figure}

\begin{figure}[!t]
\centering
  \includegraphics[width=0.45\columnwidth]{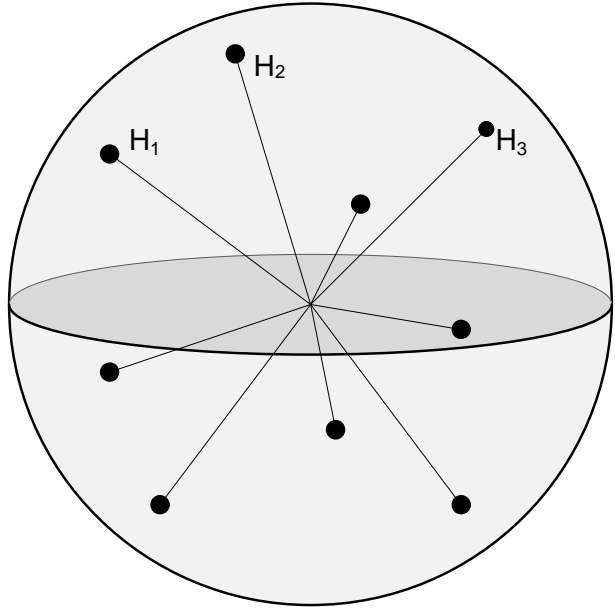}\\
  \caption{There are total $S$ points on the sphere. Each point
  represents the subspace formed by each jammer.}
  \label{fig:jamming_subspaces}
\end{figure}

\begin{figure}[!t]
\centering
\subfigure[Packing radius]
    {\includegraphics[width=.4\textwidth]{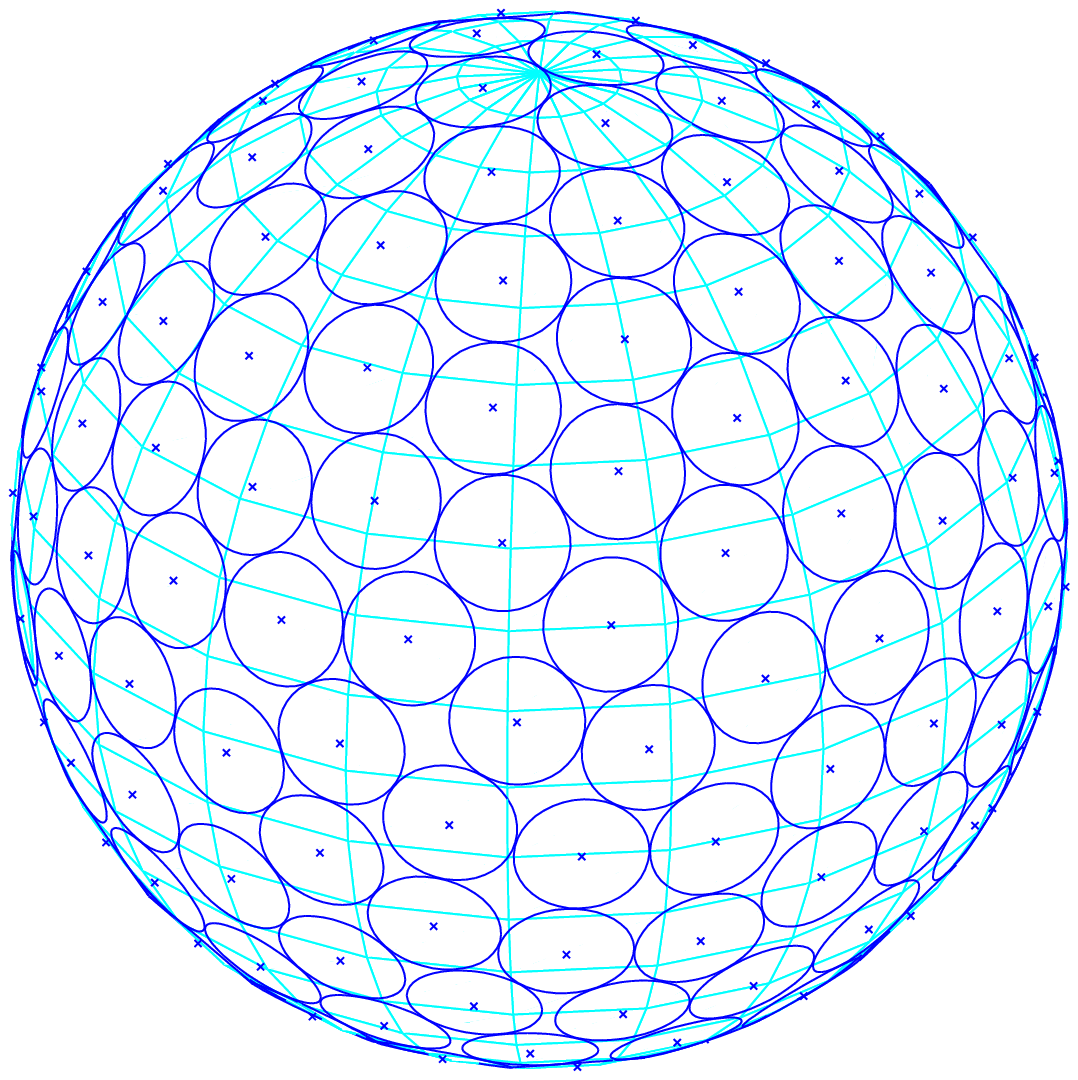}
    \label{fig:packing_radius}}
\subfigure[Covering radius]
    {\includegraphics[width=.4\textwidth]{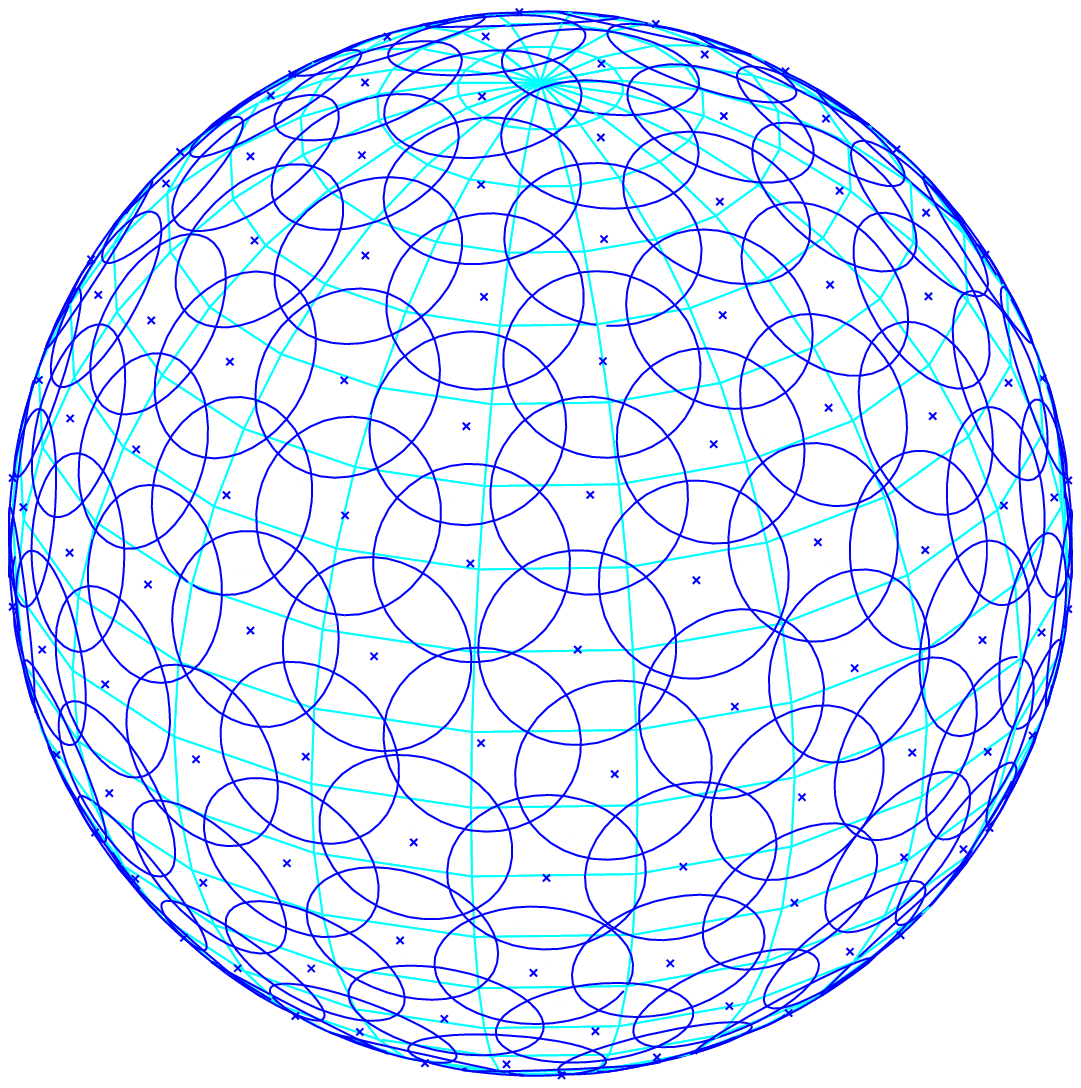}
    \label{fig:covering_radius}}
\caption{Graphical representation of the packing radius and the
covering radius on the sphere. The metric balls at $M$ points pack
or cover the sphere. } \label{fig:packing_covering}
\end{figure}

%\begin{comment}
\begin{figure}[!t]
\centering
  \includegraphics[width=.5\columnwidth]{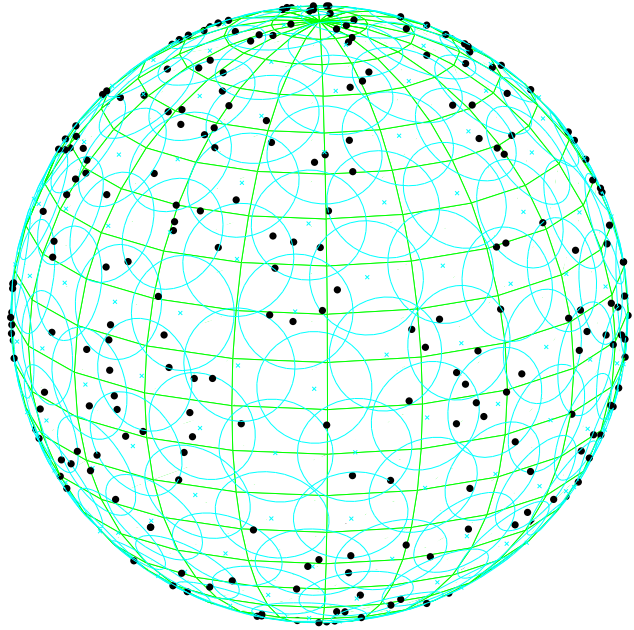}\\
%%  \caption{Jamming channels of $S=(K-1)M+1$ jammers and covering codebook of size
%%  $M$. At least $K$ jamming channels are aligned in the same metric ball.}
  \caption{There are $S=(K-1)M+1$ jamming subspaces (points) and $M$ metric balls covering the sphere.
  There exist a metric ball that contains $K$ jamming subspaces. }
  \label{fig:jammer_covering}
\end{figure}
%\end{comment}

\begin{figure}[!t]
\centering
  \includegraphics[width=.6\columnwidth]{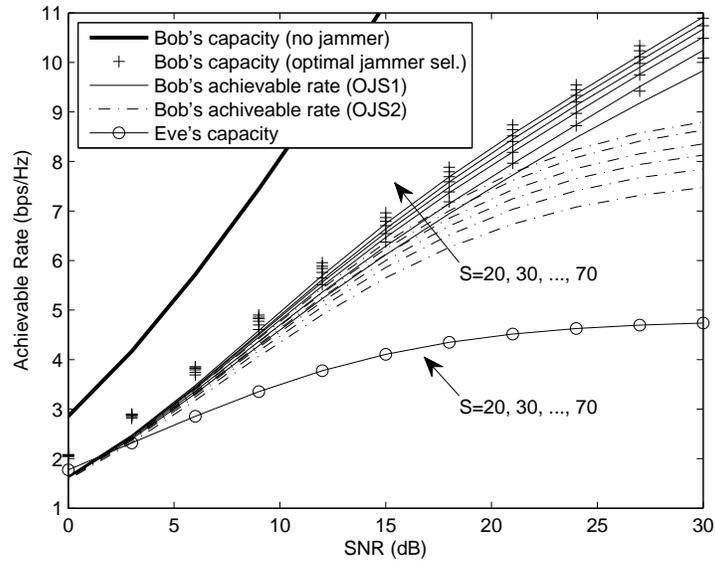}\\
  \caption{Achievable rates of Bob when the number of jammers is fixed.
  $(N_t, N_j, N_r, N_e) = (2, 2, 4, 4)$ and $K=2$.}
  \label{fig:achievable_rate}
\end{figure}

\begin{figure}[!h]
\centering
  \includegraphics[width=.6\columnwidth]
    {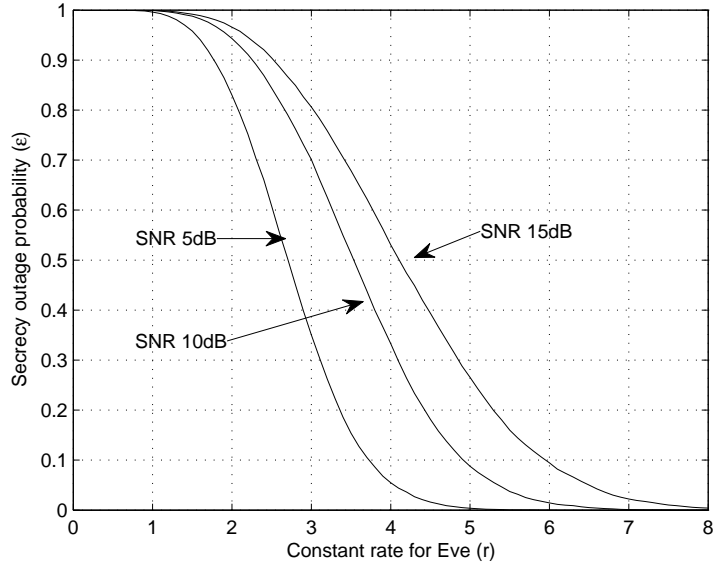}\\
  \caption{Secrecy outage probability $\epsilon$ when the Eve's achievable rate is
treated as a constant $r$. $(N_t, N_j, N_r, N_e) = (2, 2, 4, 4)$ and
$K=2$.}
  \label{fig:outage_probability}
\end{figure}

\begin{figure}[!t]
\centering
  \includegraphics[width=.6\columnwidth]{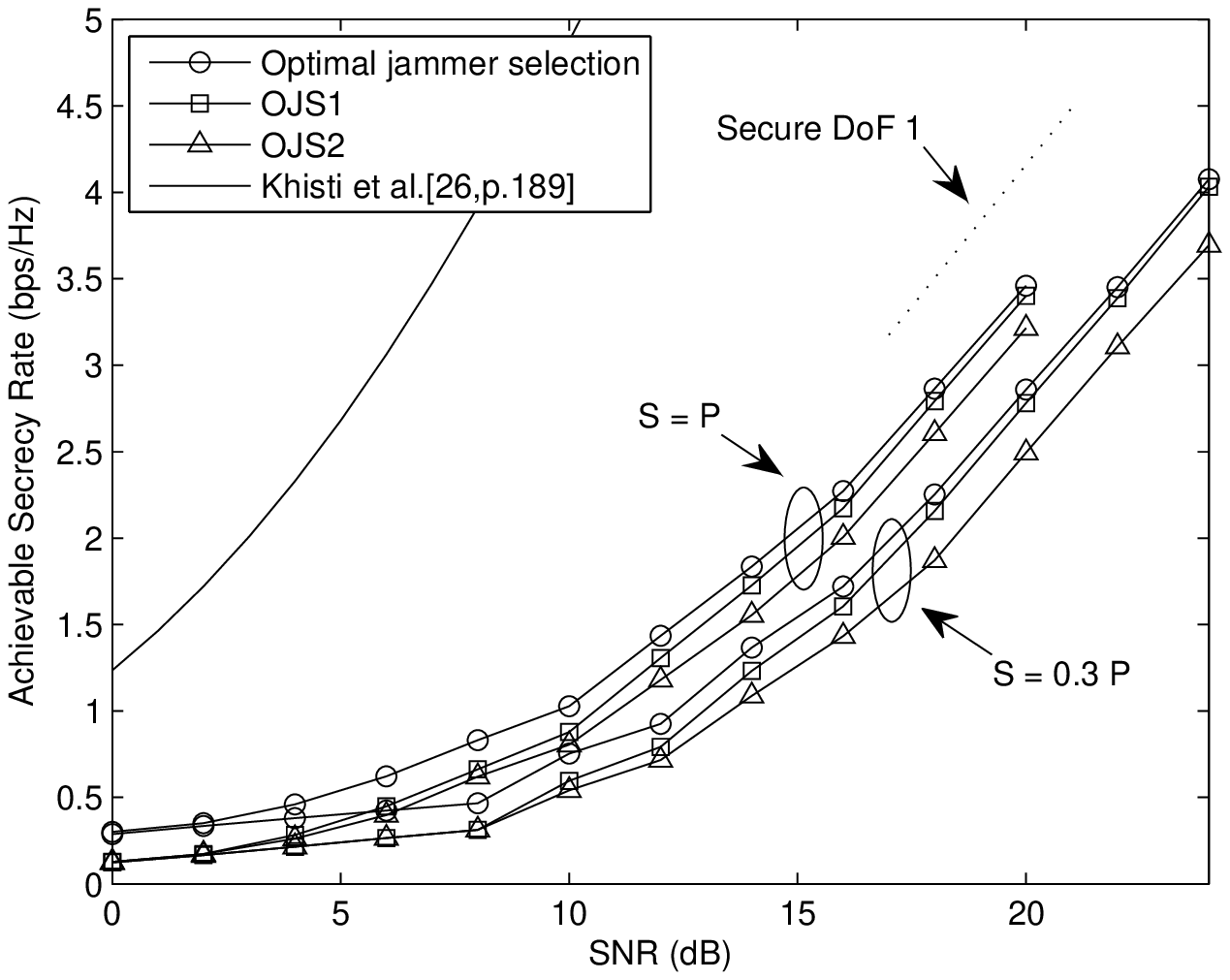}\\
  \caption{The achievable secrecy rates for various jammer selection schemes when $(N_t, N_j, N_r,
    N_e) = (1, 2, 3, 3)$. Bob selects two jammers in the jammer groups of
     $S=P$ and $S=0.3P$ jammers, respectively.}
  \label{fig:scaling1}
\end{figure}

\begin{figure}[!t]
\centering
  \includegraphics[width=.6\columnwidth]{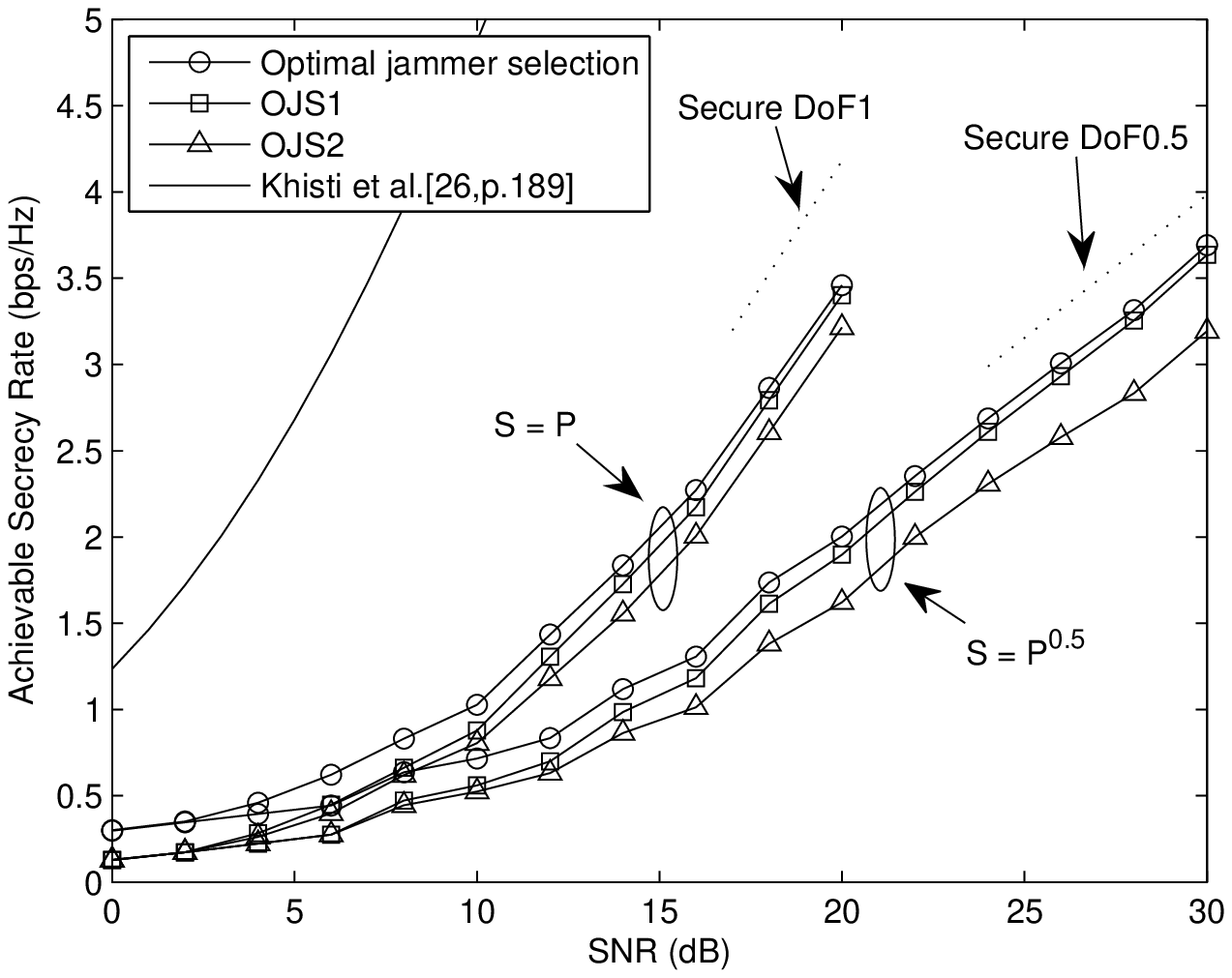}\\
  \caption{The achievable secrecy rates for various jammer selection schemes when $(N_t, N_j, N_r,
    N_e) = (1, 2, 3, 3)$. Bob selects two jammers in the jammer groups of
    $S=P$ and  $S=P^{0.5}$ jammers, respectively.}
  \label{fig:scaling2}
\end{figure}

\begin{figure}[!h]
\centering
  \includegraphics[width=.6\columnwidth]
    {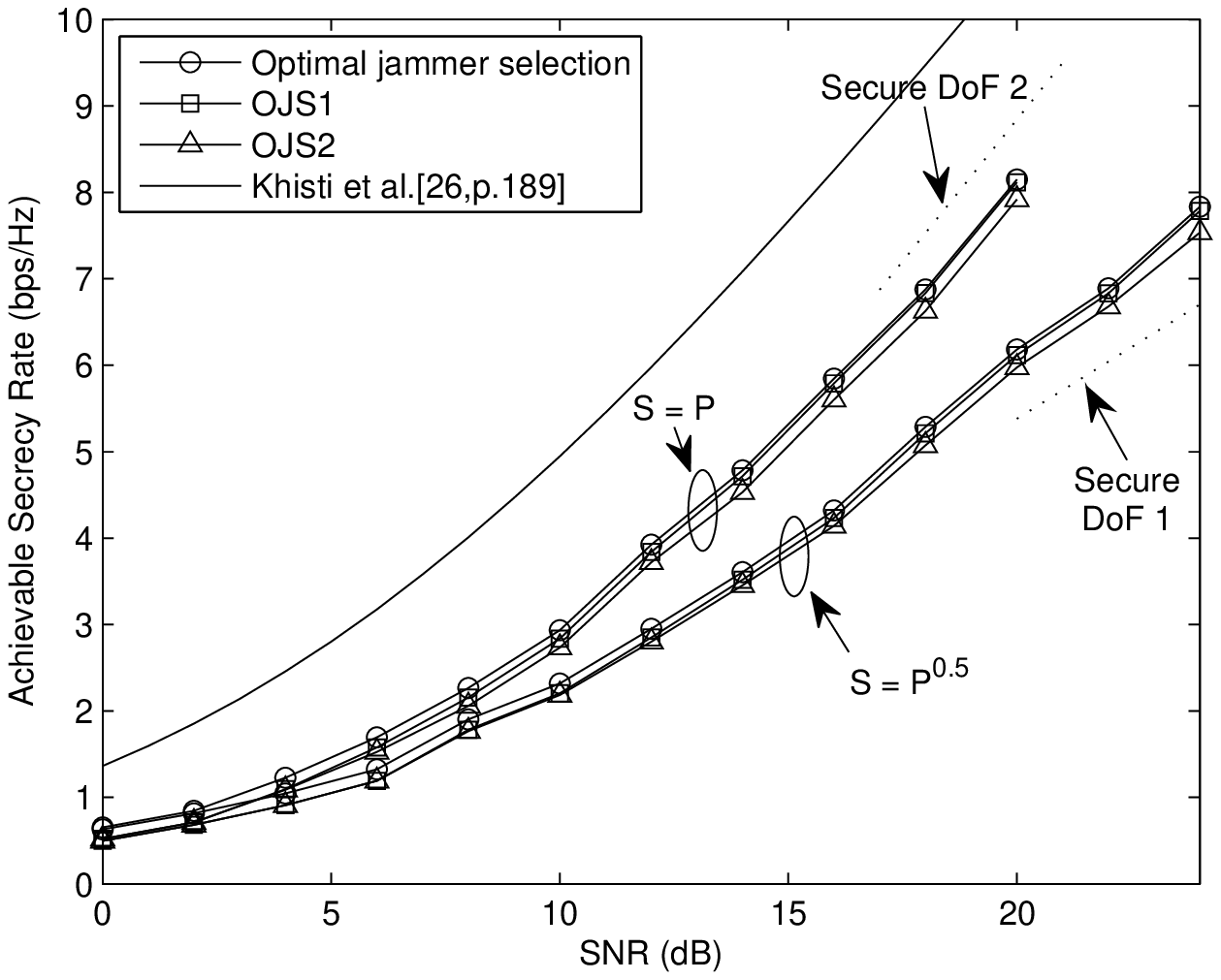}\\
  \caption{The achievable secrecy rates for various jammer selection schemes when
  $(N_t, N_j, N_r, N_e) = (2, 1, 3, 2)$. Bob selects two jammers in the jammer
  groups of $S=P$ and $S=P^{0.5}$ jammers, respectively.}
  \label{fig:scaling3}
\end{figure}

\end{document}